\newif\ifdraft\drafttrue 
\documentclass[10pt,sigplan,screen]{acmart}
\pdfoutput=1 
\newif\ifappendix\appendixtrue
\newif\ifanonymous\anonymousfalse
\usepackage{pervasives}
\usepackage{adjustbox}

\keywords{virtual machines, first-class continuations, formal language semantics}

\copyrightyear{2020}
\acmYear{2020}
\setcopyright{acmlicensed}\acmConference[DLS '20]{Proceedings of the 16th ACM SIGPLAN International Symposium on Dynamic Languages}{November 17, 2020}{Virtual, USA}
\acmBooktitle{Proceedings of the 16th ACM SIGPLAN International Symposium on Dynamic Languages (DLS '20), November 17, 2020, Virtual, USA}
\acmPrice{15.00}
\acmDOI{10.1145/3426422.3426978}
\acmISBN{978-1-4503-8175-8/20/11}

\usepackage{hyperref}
\hypersetup{%
  bookmarksdepth = {-2},
  pdftitle = {\name: Delimited Continuations for WebAssembly},
  pdfkeywords = {},
  pdfauthor = {},
  pdfstartview={FitH},
  urlcolor=black,
  linkcolor=black,
  citecolor=black,
  colorlinks=true,
}

\begin{abstract}

\term{WebAssembly} is designed to be an alternative to JavaScript that is a
safe, portable, and efficient compilation target for a variety of languages.
The performance of high-level languages depends not only on the underlying performance 
of WebAssembly, but also on the quality of the 
generated WebAssembly code. In this paper, we identify several features of high-level
languages that current approaches can only compile to WebAssembly by generating complex and
inefficient code. We argue that these problems could be addressed if
WebAssembly natively supported first-class continuations.
We then present \name, which extends WebAssembly with delimited
continuations. \name introduces no new value types, and thus does not require significant changes
to the WebAssembly type system (validation). \name is safe,
even in the presence of
foreign function calls (e.g., to and from JavaScript). Finally, \name is amenable
to efficient implementation: we implement \name as a local change to Wasmtime, an
existing WebAssembly JIT.
We evaluate \name by implementing \cname, which adds delimited continuations to C/C++. \cname
uses Emscripten and its implementation serves as a case study on how to
use \name in a compiler that targets WebAssembly. We present 
several case studies using \cname, and show that on implementing green threads,
it can outperform the state-of-the-art approach Asyncify with an $18\%$
improvement in performance and a $30\%$ improvement in code size.
\end{abstract}

\begin{document}

\ccsdesc[500]{Software and its engineering~Coroutines}

\title{\name: Delimited Continuations for WebAssembly}

\ifanonymous
\author{Anonymous Author(s)}
\else
\author{Donald Pinckney}
\email{pinckney.d@northeastern.edu}
\affiliation{\institution{Northeastern University}\country{USA}}
\author{Arjun Guha}
\email{a.guha@northeastern.edu}
\affiliation{\institution{Northeastern University}\country{USA}}
\author{Yuriy Brun}
\email{brun@cs.umass.edu}
\orcid{0000-0003-3027-7986}
\affiliation{\institution{University of Massachusetts Amherst}\country{USA}}
\fi
\maketitle

\section{Introduction}
\label{sec:intro}

For decades, ECMAScript (JavaScript) was the only programming language that was
universally supported by all major web browsers.
There are now several, high-performance JavaScript implementations that make
it possible to run large programs, such as spreadsheets, IDEs, and video editors
on the Web.
In fact, many contemporary desktop applications, such as Slack and
Visual Studio Code, are now built with JavaScript and other web technologies~\cite{wasm:electron}.

Since web browsers, and thus JavaScript, are ubiquitous, there are now
scores of programming languages with compilers that emit JavaScript to
run on the Web. However, compiling to JavaScript has two serious drawbacks:
1)~programs may perform poorly when compiled to JavaScript, and 
2)~a variety of language features, such as threads, are hard to compile to
JavaScript.

However, there is now an alternative to JavaScript.
\term{WebAssembly}~\cite{haas:2017:bringing} is a 
recently introduced low-level language
that aims to be a better compiler target language than JavaScript. All modern
web browsers support WebAssembly, and despite its name, there are several
WebAssembly runtime systems that are not embedded in browsers.
When programs written in C/C++ are compiled to WebAssembly, they run $1.3\times$
faster on average than when they are compiled to JavaScript~\cite{haas:2017:bringing, jangda:2019:not-so-fast}.
However, given WebAssembly as it exists today, it remains difficult to compile
a variety of language features, including green threads, coroutines, and continuations.
In fact, many languages that support these features natively, either do not
support them in WebAssembly, or produce slow code.
For example, the Go compiler has a WebAssembly backend. However,
it struggles to support green threads (\emph{Goroutines}), which makes the compiler
difficult to maintain, and produces code that performs 
poorly~\cite{wasm:go-issue-1, wasm:go-issue-2, wasm:go-issue-3, wasm:go-issue-4}.

\emph{Safety} is a key design goal of WebAssembly, which is necessary for
web browsers to run untrusted code in a trustworthy manner.
Toward this end, WebAssembly programs are isolated from the browser, and
cannot directly alter the low-level state of the WebAssembly runtime.
In particular, the WebAssembly stack is not stored on the WebAssembly heap.
Moreover, WebAssembly only supports structured control-flow and does not
support 
exceptions\footnote{There exists a formal proposal to extend WebAssembly
with exceptions~\cite{wasm:exceptions}.}$\!\!$, \emph{goto}, and
\emph{longjmp}.
These restrictions make WebAssembly validation straightforward and fast.
However, they make it difficult to implement non-local control flow.
For example, Goroutines and green threads require low-level support
for switching between stacks, which WebAssembly does not directly support.

\paragraph{State of the art.}

To workaround the restrictions of WebAssembly, the 
 Go compiler performs a global program transformation, which 
1)~builds a copy of the WebAssembly stack in the heap to store local variables,
and
2)~simulates non-local jumps via an elaborate state machine in each function.
Asyncify~\cite{wasm:asyncify} uses a similar approach to add virtual
instructions that save and restore stacks to WebAssembly.
 Prior benchmarks indicate that
Asyncify has a performance overhead
of 20\%--100\%, and a similar increase in code size~\cite{wasm:asyncify}.
Moreover, since these tools use a \emph{global transformation} to achieve non-local control flow,
programmers are forced to pay a steep cost for such features, even when their code
uses them minimally.

\paragraph{Our contributions.}

In this paper, we present \name, an extension to WebAssembly that adds support
for \emph{first-class, delimited continuations},
which are sufficient to implement a wide variety of language features,
including green threads, coroutines, and exceptions. Our extension has only
a handful of new instructions, is designed to support efficient implementation,
and is designed to work well when WebAssembly programs interact with other
languages (e.g., JavaScript).

Although first-class continuations are a well-known abstraction, they are
typically found in higher-level programming languages
(e.g., Scheme and Racket). These
languages are compiled to low-level native code that does not support
first-class continuations. Our work inverts this tradition and instead
adds first-class continuations directly to a low-level language.
In doing so, our design tackles unique challenges imposed by the low-level setting,
such as the lack of first-class functions, lack of garbage collection, 
and the requirement that WebAssembly support safe interoperability with host
languages (e.g., JavaScript).

Another goal of our work is to ensure that our new instructions align with
with WebAssembly's performance, portability, and safety objectives. 
We considered the following design goals: 
1)~Common language features, such as green threads, 
should be able to compile to efficient \name code. 
2)~The extension should lend itself to simple
type checking (validation).
3)~The extension should lend itself to high-performance implementation in
existing WebAssembly JITs.
4)~The extension must be safe.
5)~Existing WebAssembly instructions and code should suffer no performance penalty. And
6)~the performance of new instructions should be fast and predictable.

Since the goal of \name is to provide a better 
compiler target language, we also prototyped \cname, 
an extension to C/C++ that adds support for
delimited continuations. \cname uses the Emscripten compiler that
compiles from C/C++ to \name, and we use it to implement programs with a variety of
features, including green threads. We then evaluate the performance
of green threads when implemented with \cname against the state-of-the-art
approach Asyncify~\cite{wasm:asyncify}.

To summarize, we make the following contributions: 

\begin{enumerate}

\item We design \name, and present its semantics, validation, and safety
properties.

\item We implement \name as a modest extension to Wasmtime, which is a 
real-world WebAssembly JIT.

\item Using \name and the Emscripten compiler from C/C++ to WebAssembly, we
present \cname, which adds delimited continuations to C/C++.

\item We evaluate the performance of our \name implementation by comparing
it to a third-party tool that implements continuations by source-to-source
transformation.

\end{enumerate}

\definecolor{floatColor}{rgb}{1.0, 0.9, 0.85} 
\definecolor{barColor}{rgb}  {0.85, 0.9, 1.0} 

\begin{figure*}[ht]
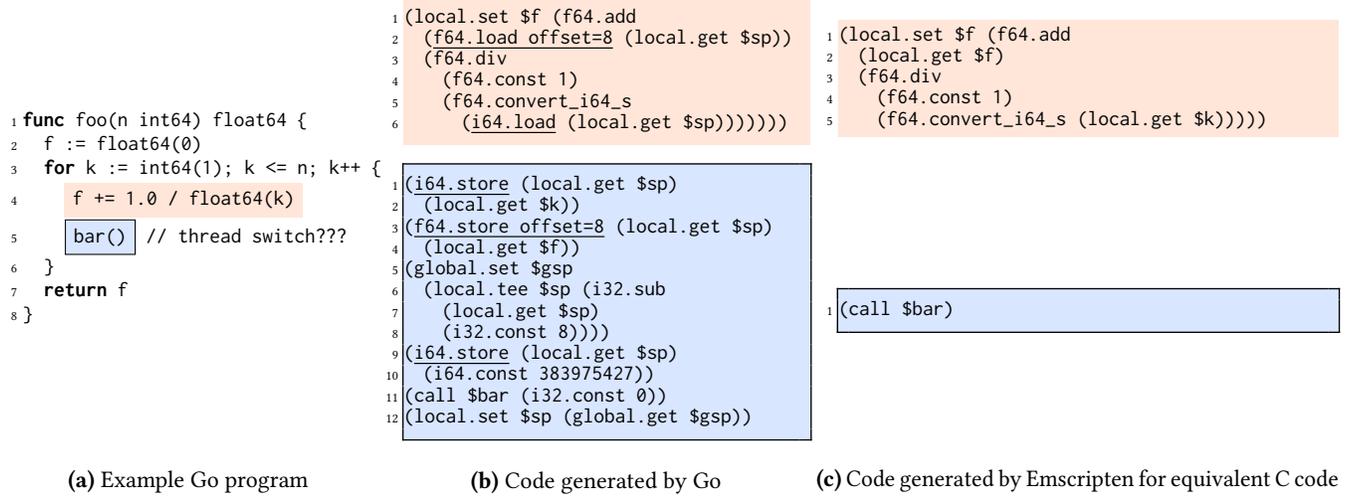

\begin{subfigure}{.28\textwidth}
\vspace{4em}
\lstset{language=MyGo}
\begin{lstlisting}
func foo(n int64) float64 {
  f := float64(0)
  for k := int64(1); k <= n; k++ {
    |\colorbox{floatColor}{f += 1.0 / float64(k)}|
    |\fcolorbox{black}{barColor}{bar()}| // thread switch???
  }
  return f
}
\end{lstlisting}
\vspace{4.3em}
\caption{Example Go program}
\label{fig:case-go-src}
\end{subfigure}
\begin{subfigure}{.32\textwidth}
\lstset{language=Wasm,frame=tb,framesep=4pt,framerule=0pt,backgroundcolor=\color{floatColor}}
\begin{lstlisting}
(local.set $f (f64.add 
  (|\underline{f64.load offset=8}| (local.get $sp)) 
  (f64.div 
    (f64.const 1) 
    (f64.convert_i64_s 
      (|\underline{i64.load}| (local.get $sp)))))))
\end{lstlisting}

\lstset{language=Wasm,frame=tlbr,framesep=4pt,framerule=0.5pt,backgroundcolor=\color{barColor},framexleftmargin=-4pt,framexrightmargin=-4pt}
\begin{lstlisting}
(|\underline{i64.store}| (local.get $sp) 
  (local.get $k))
(|\underline{f64.store offset=8}| (local.get $sp) 
  (local.get $f))
(global.set $gsp 
  (local.tee $sp (i32.sub 
    (local.get $sp) 
    (i32.const 8))))
(|\underline{i64.store}| (local.get $sp) 
  (i64.const 383975427))
(call $bar (i32.const 0))
(local.set $sp (global.get $gsp))
\end{lstlisting}
\caption{Code generated by Go}
\label{fig:case-go-bin}
\end{subfigure}
\begin{subfigure}{.39\textwidth}
\vspace{1em}
\lstset{language=Wasm,frame=tb,framesep=2pt,framerule=0pt,backgroundcolor=\color{floatColor}}
\begin{lstlisting}[escapechar=\%]
(local.set $f (f64.add
  (local.get $f)
  (f64.div
    (f64.const 1)
    (f64.convert_i64_s (local.get $k)))))
\end{lstlisting}
\vspace{5em}
\lstset{language=Wasm,frame=tlbr,framesep=4pt,framerule=0.5pt,backgroundcolor=\color{barColor},framexleftmargin=-4pt,framexrightmargin=-4pt}
\begin{lstlisting}
(call $bar)
\end{lstlisting}
\vspace{4em}
\caption{Code generated by Emscripten for equivalent C code}
\label{fig:case-emcc-bin}
\end{subfigure}
\caption{Current WebAssembly code generation by Go and Emscripten.} 
\label{fig:case-study-gen}
\end{figure*}

\section{The \name Approach}
\label{sec:tour}


We start by illustrating the basics of WebAssembly and sketch the compilation
strategy that the Go compiler uses to support Goroutines in WebAssembly
(\cref{sec:WebAssembly}). We then present \name, which extends WebAssembly
with first-class continuations (\cref{sec:approach overview}), and \cname,
which extends C/C++ with first-class continuations (\cref{sec:c-k}). We use
\cname to present green threads, generators, and probabilistic programming in
WebAssembly (\cref{sec:using-wasmk}).

%

\subsection{WebAssembly}
\label{sec:WebAssembly}

WebAssembly is a stack machine with a conceptually independent control
stack and value stack. For example, the
$\op[2]{i64.const}$ instruction pushes the integer $2$ onto the value stack,
and the $\op{i64.mul}$ instruction pops two integers off the stack, and pushes
their product onto the stack. Similarly, functions receive their arguments
and return their result on the value stack too. Each function has a collection
of local variables, and can use $\op[\$x]{local.get}$ to push a local variable's
or an argument's value onto the value stack, and $\op[\$x]{local.set}$ to pop a value
off the stack and update the variable. Similarly to local variables in C, registers
will be allocated for local variables during JIT compilation.
For example, the \texttt{\$quadruple} function (\cref{ex-quadruple-call})
quadruples its argument by calling \texttt{\$helper} to first double the
argument, and then doubles the result produced by \texttt{\$helper}.

In addition to storing data on the stack and in local variables,
data can also be stored in \term{linear memory} (WebAssembly's heap), 
which is a byte-addressable region of memory. Linear memory can be read
and written via $\op{i32.load}$ and $\op{i32.store}$ respectively, 
for, e.g., the $\mathbf{i32}$ type. Unlike local variables, using linear memory
will always incur hardware loads and stores.

\paragraph{Compiling Go to WebAssembly.}

\emph{Goroutines}, which are similar to green threads (or, user-space threads),
are the primary concurrency abstraction of the Go programming language.
When compiled to native code, the Go runtime manages a
pool of physical threads and uses them to run several Goroutines on a single thread 
(a so-called $M$:$N$ threading model). This
involves using low-level instructions to save and restore the stack and
registers' values, during a (user-space) context-switch from one Goroutine
to another. A context-switch may occur for a number of reasons, e.g., when
the active Goroutine is blocked on I/O or simply periodically.

While physical threads are discussed in the threading proposal~\cite{wasm:threads}
and are available in some WebAssembly runtimes (e.g. Chrome),
switching between goroutines within a single physical thread in WebAssembly 
is far more difficult compared to native code, since
the WebAssembly stack is not allocated in linear memory and thus
cannot be saved or restored.
\footnote{This design ensures that a malicious program cannot alter
return addresses to escape the WebAssembly sandbox.}
Instead, the compiler generates code that maintains a
heap-allocated copy of the stack residing in linear memory.

\cref{fig:case-go-src} shows an example of a simple function \texttt{foo} in Go.
The function computes a partial sum of a series, and calls
the function \texttt{bar} on each iteration. Since \texttt{bar} may trigger
a thread switch (Go inserts a yield at the start of every function), 
the compiler has to generate code to save and restore \texttt{foo}'s stack.

We sketch the generated code in \cref{fig:case-go-bin}.
Before the call to \texttt{bar}, the generated code saves the local variables
(\texttt{f} and \texttt{k}) onto the copy of the stack in linear memory 
(bottom of \cref{fig:case-go-bin}).
Thus if \texttt{bar} switches to a new Goroutine,
the local variables of \texttt{foo} can be safely discarded.
Conversely, \texttt{foo} reads the values of its local variables from the
heap-allocated stack (top of \cref{fig:case-go-bin}). Note that no such
load or store instructions need to be emitted when compiling Go to native code
directly, because in native code the Go runtime can freely manipulate the
machine stack.

It is instructive to consider how code generation works for simpler
languages, such as C. Given the C equivalent of \texttt{foo},
the Emscripten compiler from C to WebAssembly generates much simpler
code (\cref{fig:case-emcc-bin}), without any loads and stores to linear memory.
Since the code uses WebAssembly local variables exclusively, a WebAssembly
JIT can easily allocate them to machine registers. (Emscripten
does not support \texttt{setjmp} and \texttt{longjmp}, which can be used to
build green threads in C.)

The additional loads and stores in Go have a cost. In a call to \texttt{foo(}$2^{30}$\texttt{)}, 
the \texttt{perf} tool shows the Go program executes
$2.5\times$ more instructions, $3.0\times$ more branches, 
$1.9\times$ more loads, and $1.5\times$ more stores than the equivalent
C program, when we compile both to WebAssembly (compiled with optimizations and run with \texttt{node v14.4.0}).
Overall, the Go program takes $1.8\times$  longer than the C
program. However, when compiled to native code, the performance of the C and Go code is 
nearly identical.

\subsection{\name}
\label{sec:approach overview}

\begin{figure}
	\centering
	\begin{subfigure}{0.48\textwidth}
	\lstset{language=Wasm}
	\begin{lstlisting}
(func $helper (param $x i64) (result i64)
    local.get $x
    i64.const 2
    i64.mul)
(func $quadruple (param $x i64) (result i64)
    local.get $x
    call $helper
    i64.const 2
    i64.mul)
	\end{lstlisting}
	\caption{This code doubles its input first by calling a helper function, and then doubling again.}
	\label{ex-quadruple-call}
	\end{subfigure}
	\begin{subfigure}{0.48\textwidth}
	\lstset{language=Wasm}
	\begin{lstlisting}
(func $handler (param $k i64) (param $x i64)
    local.get $k
    local.get $x
    i64.const 2
    i64.mul
    restore)
(func $quadruple2 (param $x i64) (result i64)
    local.get $x
    control $handler
    i64.const 2
    i64.mul)
	\end{lstlisting}
    \caption{In contrast, this code captures the current stack, and then jumps back to that stack.}
	\label{ex-quadruple-control}
	\end{subfigure}
\caption{Two ways to write a function which quadruples its input.}
\label{ex-quadruple}
\end{figure}

\name adds five new instructions to WebAssembly. 
1)~The $\control{h}$ instruction captures the current continuation, stores
it a region of memory called the \term{continuation table},
assigns it a new \emph{continuation ID} ($\kappa$),
and invokes the function ($h$) with a fresh stack,
passing the continuation ID and a user-provided argument.
2)~$\restore$ receives a continuation ID as its argument, and restores
the associated continuation, discarding the current continuation in the process.
It is a runtime
error to $\restore$ the same continuation multiple times. 
3)~The $\continuationcopy$ instruction creates a copy of a continuation.
4)~The $\continuationdelete$ instruction deletes a continuation without
restoring it. 5)~The $\prompt{\!}{e^*}$ instruction wraps a block of
instructions ($e^*$), and serves as a delimiter for continuation capture:
all continuations captured by $e^*$ do not extend beyond the call to prompt.

\cref{ex-quadruple-control} shows an alternate implementation of \texttt{\$quadruple},
using continuations in a trivial way. The function captures
its continuation and passes it to \texttt{\$handler} (line 9), which runs in the
empty continuation. The \texttt{\$handler} function receives two arguments: \texttt{\$k} is 
the captured continuation ID, and \texttt{\$x} is the 
original argument to \texttt{\$quadruple2}, passed through to \texttt{\$handler}.
The \texttt{\$handler} function doubles its argument and restores \texttt{\$k},
passing the doubled value along. At this point, the captured continuation will execute
(line 10), with the doubled valued pushed onto the stack. 
Execution completes as before by doubling again.

\begin{figure}[t]
\lstset{language=MyC}
\begin{lstlisting}
typedef uint64_t k_id;
typedef void (*control_handler_fn)(k_id, uint64_t);

uint64_t control(uint64_t arg, control_handler_fn fn_ptr);
void restore(k_id k, uint64_t val);
uint64_t continuation_copy(k_id k);
void continuation_delete(k_id k);
#define prompt(x) <...>
\end{lstlisting}
\caption{A C/C++ First-Class Continuations Header}
\label{c-function-proto}
\end{figure}
    
\subsection{\cname: Continuations for C/C++}
\label{sec:c-k}

Writing and reading substantial examples in WebAssembly is tedious. Therefore,
the rest of this paper presents examples using C/C++. We use the Emscripten compiler
from C/C++ to WebAssembly, and export the new \name instructions to C/C++ programs
using the API defined in \cref{c-function-proto}. Each of these functions
call their corresponding instructions in \name to manipulate the WebAssembly stack.

However, it is not enough to directly expose the \name primitives to C++ code.
A program written in C/C++ can get the memory address of a local variable,
which WebAssembly does not support.
Emscripten uses a heap-allocated portion of the stack to support these programs.
Therefore, \cname has to carefully manage this portion of the stack as well (\cref{sec:in-practice}).

Adding first-class continuations in this manner to C/C++ is unusual, 
as typically first-class continuations are a feature in
high-level languages such as Racket, and need to be compiled to low-level code
which does not support first-class continuations. By going in the 
opposite direction, we get them almost for free in a higher-level language.

\subsection{Using Continuations in \name and \cname}
\label{sec:using-wasmk}

We now present several applications of \name, using \cname to write our code.

\paragraph{Green threads.}

\begin{figure}
\begin{subfigure}{\columnwidth}
\lstset{language=MyC}
\begin{lstlisting}
void thread_main() {
    std::cout << "A" << std::endl;
    thread_yield();
    std::cout << "B" << std::endl;
}
int main() {
    thread_create(thread_main);
    thread_create(thread_main);
    join_all_threads();
}
\end{lstlisting}
\caption{Example of use. Prints \texttt{AABB}}
\label{ex-user-threads-usage}
\end{subfigure}

\begin{subfigure}{\columnwidth}
\lstset{language=MyC}
\begin{lstlisting}
std::vector<uint64_t> Q;
uint64_t after_join;
uint64_t dequeue() {
    uint64_t next_k = Q.back(); Q.pop_back(); 
    return next_k;
}

void save_fk_restore(uint64_t fk, uint64_t create_k) {
    restore(create_k, fk);
}
void create_handler(uint64_t k, uint64_t f) {
    control(save_fk_restore, k);
    ((void (*)())f)();
    if(Q.size() > 0) {
        restore(dequeue(), 0);
    } else {
        restore(after_join, 0);
    }
}
void thread_create(void (*f)()) {
    Q.insert(Q.begin(), control(create_handler, (uint64_t)f));
}

void join_handler(uint64_t k, uint64_t arg) {
    after_join = k;
    restore(dequeue(), 0);
}
void join_all_threads() {
    control(join_handler, 0);
}

void yield_handler(uint64_t k, uint64_t arg) { 
    Q.insert(Q.begin(), k); |\label{line:q_insert}|
    restore(dequeue(), 0); |\label{line:dq_restore}|
}
void thread_yield() {
    control(yield_handler, 0); |\label{line:yield_control}|
}
\end{lstlisting}
\caption{Implementation.}
\label{ex-user-threads-impl}
\end{subfigure}
\caption{Green threads in \cname}
\label{ex-wasm-user-threads}
\end{figure}

Green threads (or cooperative threads), are a simple example of an abstraction
that is easy to build with continuations. \Cref{ex-user-threads-impl} shows
an implementation of green threads in \cname, which provides functions to
create new threads, wait on threads to complete, and suspend the running
thread and yield control to another thread (\texttt{thread\_yield}).
\Cref{ex-user-threads-usage} is a small program that uses this threading 
library.

The key insight is that \texttt{thread\_yield} can be accomplished by 
capturing the current continuation of the thread via \texttt{control} (\cref{ex-user-threads-impl} \cref{line:yield_control}), 
storing the continuation ID in a queue (\cref{ex-user-threads-impl} \cref{line:q_insert}), 
and dequeue-ing and restoring another continuation ID (\cref{ex-user-threads-impl} \cref{line:dq_restore}).
Since green threads do not need to pass data between threads, we do not utilize 
the data arguments to \texttt{control} and \texttt{restore}.

\begin{figure}
\begin{subfigure}{\columnwidth}
\lstset{language=MyC}
\begin{lstlisting}
void example_generator(Generator *g) {
    uint64_t i = 0;
    while(1) { gen_yield(i++, g); }
}
int main() {
    Generator *g = make_generator(example_generator);
    for(int i = 0; i < 10; i++) 
        printf("%llu\n", gen_next(g));
    free_generator(g);
    return 0;
}
\end{lstlisting}
\caption{Example of use.}
\label{generators-example}
\end{subfigure}

\begin{subfigure}{\columnwidth}
\lstset{language=MyC}
\begin{lstlisting}
typedef struct { 
    k_id after_next, after_yield; uint64_t value; 
} Generator; |\label{line:generator_struct}|
// Helpers for converting a function to a continuation |\label{line:generator_help_start}|
void return_convert_result(uint64_t k, uint64_t ak) { 
    restore(ak, k); 
}
void convert_handler(uint64_t k, void (*f)(Generator*)) {    
    f((Generator *)control(return_convert_result, k));
}
uint64_t convertFuncToCont(void (*f)(Generator*)) { 
    return control(convert_handler, f); 
} |\label{line:generator_help_end}|
// Allocating a generator |\label{line:generator_make_start}|
Generator *make_generator(void (*f)(Generator*)) { 
    Generator *g = (Generator *)malloc(sizeof(Generator));
    g->after_yield = convertFuncToCont(f); return g; |\label{line:generator_help_use}|
} |\label{line:generator_make_end}|
// Yielding implementation |\label{line:generator_yield_start}|
void yield_handler(k_id k, Generator *g) {
    g->after_yield = k;
    restore(g->after_next, g->value); |\label{line:generator_restore_after_next}|
}
void gen_yield(uint64_t v, Generator *g) {
    g->value = v;
    control(yield_handler, g);
}
// Next implementation
void next_handler(k_id k, Generator *g) {
    g->after_next = k;
    restore(g->after_yield, 0);
}
uint64_t gen_next(Generator *g) {
    return control(next_handler, g);
} |\label{line:generator_next_end}|
// Freeing a generator |\label{line:generator_free_start}|
void free_generator(Generator *g) {
    continuation_delete(g->after_yield); free(g); |\label{line:generator_delete_use}|
} |\label{line:generator_free_end}|
\end{lstlisting}
\caption{Implementation.}
\label{generators-impl}
\end{subfigure}
\caption{Generators in \cname.}
\end{figure}

\paragraph{Generators.}

\emph{Generators} are a programming abstraction that are found in
a variety of languages, including Python and JavaScript. Although C does not
support generators, we can build them using
\texttt{control} and \texttt{restore}.
\Cref{generators-example} shows a program in \cname that prints the
numbers $0$ through $9$, using a generator function. The generator contains
what appears to be an infinite loop, but each iteration suspends execution
in the generator (\texttt{gen\_yield}) and resumes execution in
\texttt{main}.

\Cref{generators-impl} presents the implementation of generators using \cname.
The primary difference between our
implementation and canonical implementations (e.g., in Racket~\cite{wasm:racket-generator}), is
that C does not support first-class functions. Therefore,
we represent a generator as an object (\texttt{struct}) with fields that hold
1)~the ID of the continuation where the generator was invoked 
(\texttt{after\_next}), 
2)~the ID of the continuation where the generator was
last suspended (\texttt{after\_yield}), and 
3)~the next 
value to return from the generator (\texttt{value}).

Finally, the generator API includes a function to delete a generator
object (\texttt{free\_generator}).
This function deletes the continuation
within the generator (\texttt{g->after\_yield}) using \texttt{continuation\_delete} 
(\cref{line:generator_delete_use}).
Note that since the other continuation (\texttt{g->after\_next}) was restored to during
the most recent yield (\cref{line:generator_restore_after_next}), it is currently 
unallocated and does not need to be deleted.
The need for a $\continuationdelete$ instruction is subtle, but is required for
natural use cases of first-class continuations in a low-level language without
garbage collection.

\paragraph{Probabilistic programming.}

\begin{figure}
\begin{subfigure}{\columnwidth}
\lstset{language=MyC}
\begin{lstlisting}
uint64_t sum_d6() {
    auto *d6 = new std::vector<uint64_t> {1, 2, 3, 4, 5, 6};
    return uniform(d6) + uniform(d6);
}
int main() {
    std::cout << *driver(sum_d6) << std::endl; return 0;
}
\end{lstlisting}
\caption{Example of use.}
\label{ex-prob-use}
\end{subfigure}
\begin{subfigure}{\columnwidth}
\lstset{language=MyC}
\begin{lstlisting}
struct ContinuationThunk { |\label{line:thunk_start}|
    k_id continuation; // The continuation to resume
    uint64_t value; // The value to pass to the continuation
};
// vector of thunks which need to be executed
std::vector<ContinuationThunk *> to_execute; |\label{line:thunk_end}|

std::map<uint64_t, double> *driver(uint64_t (*body)()) { |\label{line:driver_start}|
    auto *results = new std::vector<uint64_t>();
    results->push_back(body()); |\label{line:prob_push_back}|
    if(rest.size() > 0) {
        ContinuationThunk *t = rest.back(); rest.pop_back();
        restore(t->continuation, t->value); |\label{line:prob_restore}|
    }
    return count_probs(results);
} |\label{line:driver_end}|

void uniform_handler(k_id k, std::vector<uint64_t> *args) {    
    for(auto it = std::next(args->begin()); 
        it != args->end(); ++it) {
        to_execute.push_back(new ContinuationThunk {
            .continuation=continuation_copy(k), |\label{line:uniform_copy}|
            .value=*it});
    }
    restore(k, args[0]); |\label{line:uniform_restore}|
}
uint64_t uniform(std::vector<uint64_t> *args) {
    return control(uniform_handler, args); |\label{line:uniform_control}|
}
\end{lstlisting}
\caption{Implementation.}
\label{ex-prob-impl}
\end{subfigure}
\caption{An embedded probabilistic programming language in C++.}
\end{figure}

A more involved example is the implementation of 
an embedded probabilistic programming language in C++. 
Probabilistic programming languages allow probabilistic models
to be implemented declaratively in general purpose languages.
One common approach to implement a probabilistic programming language
is to relate sampling from a probability 
distribution to sampling from a distribution of program executions~\cite{paige:2014:compilation}.
Performing this sampling requires some use of control 
operators which can essentially fork execution
to allow it to be re-executed (i.e., sampled from) multiple times. While the
implementation of a proper probabilistic programming language with modern sampling
algorithms is out of the scope of this paper, we can nevertheless demonstrate how to
implement a probabilistic programming language allowing for finite distributions 
embedded in C++.

An example usage of the embedded probabilistic programming language is shown 
in \cref{ex-prob-use}. 
\texttt{sum\_d6} computes the sum of two independent dice rolls.
The call to \texttt{driver(sum\_d6)} will run the sampling algorithm, 
eventually returning a map that represents
the probability mass function (PMF) of \texttt{sum\_d6}. 
The proposed API consists of just a \texttt{uniform} 
function which represents the uniform distribution over a discrete set of values 
(the vector argument) and the \texttt{driver} function which conducts the sampling 
to obtain the final PMF. This API can be easily expanded in this framework to allow for
different distributions and conditioning, but these are omitted for brevity.

The implementation of the API is shown in \cref{ex-prob-impl}.
The core idea is that each sample from a distribution will correspond to forking 
the execution for each sampled value. For example, if sampling from 
\texttt{uniform(1, 2, 3)} the execution would be forked into 3 executions, 
one with each sampled value. The various execution forks are stored in the 
\texttt{to\_execute} state, in the format of a vector of 
\texttt{Continuation\-Thunk}s (lines \ref{line:thunk_start}--\ref{line:thunk_end}), which keep track of the continuation to 
restore to, and the sampled value to pass to the continuation upon restoring. 

The implementation of \texttt{driver} (lines \ref{line:driver_start}--\ref{line:driver_end}) keeps a vector of final sampled values, 
and proceeds by first running the given function argument (\texttt{body()}) and saving 
the result, and then dequeuing a thunk to execute and restoring it. 
Supposing that \texttt{body} forked its execution into thunks, then the call to
\texttt{restore} (\cref{line:prob_restore}) will jump back into the execution of somewhere in \texttt{body},
eventually returning yet again to the \texttt{push\_back} (\cref{line:prob_push_back}).
Thus, \texttt{driver} will continue to push results and dequeue a new thunk, until all
thunks (samples) are exhausted. Finally, \texttt{count\_probs} computes the
desired map.

With \texttt{driver} worked out, the implementation of \texttt{uniform} is 
conceptually straightforward: \texttt{uniform(args)} should fork the execution for 
each value in \texttt{args}. This is accomplished by first immediately 
calling \texttt{control} (\cref{line:uniform_control}) to capture the current continuation. Then, for every
element except the first element of \texttt{args} a new thunk is queued, where the
continuation is a \emph{copy} of the current continuation \texttt{k} (\cref{line:uniform_copy}). 
An explicit copy of \texttt{k} is required because all of these thunks will eventually
be restored to, and under one-shot continuation semantics it is invalid to restore
to a single continuation (\texttt{k}) multiple times. Finally, the current continuation
is restored immediately with the first sampled value rather than saved in a thunk
(\cref{line:uniform_restore}).

\section{Semantics of \name}
\label{sec:approach}

This section presents 1)~an overview of WebAssembly's operational semantics,
2)~extends the operational semantics to support continuations, 3)~presents
type-checking (known as validation) for this extension, and 4)~proves that the
extension is sound.

\begin{figure*}
\footnotesize
\begin{subfigure}{0.4\textwidth}
\[
\begin{array}{@{}ll}
\op[2]{i64.const} & \wasmcomment{push 2} \\
\op{block} & \wasmcomment{enter block} \\
\op[3]{i64.const} & \wasmcomment{push 3} \\
\op[4]{i64.const} & \wasmcomment{push 4} \\
\op{i64.add} & \wasmcomment{pop 3 \& 4, push 7} \\
\op[0]{br} & \wasmcomment{exit block}\\
\op{end} & \wasmcomment{end of block} \\
\op{i64.sub} & \wasmcomment{pop 2 \& 7, push $-5$}
\end{array}
\]
\caption{Example program.}
\label{example-wasm-stack}
\end{subfigure}
\begin{subfigure}{0.59\textwidth}
\[
\begin{array}{@{}l@{\;}l}
  & \op[2]{i64.const} \; \op{block} \; \op[3]{i64.const} \; \op[4]{i64.const} \; \op{i64.add} \; \op[0]{br} \; \op{end} \; \op{i64.sub} \\
\hookrightarrow  & \op[2]{i64.const} \; \mathbf{label}\{\epsilon\} \; \underbrace{\op[3]{i64.const} \; \op[4]{i64.const} \; \op{i64.add} \; \op[0]{br}}_\textrm{within an $L^1$ context} \; \op{end} \; \op{i64.sub} \vspace{0.5em} \\
\hookrightarrow  & \op[2]{i64.const} \; \mathbf{label}\{\epsilon\} \; \op[7]{i64.const} \; \op[0]{br} \; \op{end} \; \op{i64.sub} \\
\hookrightarrow	& \op[2]{i64.const} \; \op[7]{i64.const} \; \op{i64.sub} \\
\hookrightarrow & \op[(-5)]{i64.const}
\end{array}
\]
\caption{Reduction sequence.}
\label{example-wasm-red-seq}
\end{subfigure}

\caption{An example of WebAssembly execution.}
\end{figure*}

\subsection{WebAssembly Semantics}

WebAssembly is formalized as a stack-based, small-step reduction semantics.
This section introduces a small fragment of the WebAssembly semantics,
using the example program in \cref{example-wasm-stack}.
For a more detailed account, we refer the reader to Haas et al.~\cite{haas:2017:bringing}
and the WebAssembly specification~\cite{wasm:spec}.

The WebAssembly stack machine contains both instructions ($e$) that are pending
evaluation and values ($v$) that were produced by instructions that have
already been evaluated. The values are a subset of instructions.
For example, the instruction $\op[n]{i64.const}$ pushes the 64-bit value $n$ 
onto the stack. Moreover, the semantics represents the 64-bit value $n$ as
$\op[n]{i64.const}$.  In the absence of control flow and function calls, 
a configuration of the WebAssembly stack machine has a sequence of
evaluated values ($v^*$) and a sequence of instructions ($e^*$)
in succession ($v^* e^*$), and we
always evaluate the first non-value instruction in the sequence.
The boundary between values and instructions at which evaluation occurs
is called the \term{local context of depth 0} ($L^0[\_]$).

WebAssembly has \term{structured control flow}, and does not have 
goto-style instructions. Instead, the language has structured control
flow blocks (e.g., $\op{block} \ldots \op{end}$ and $\op{loop} \ldots \op{end}$).
The WebAssembly semantics turns all kinds of blocks into labelled blocks
($\mathbf{label}\{e^*\} \ldots \op{end}$), which are an
administrative instruction.\footnote{The
$e^*$ is only needed to encode loops, and can be ignored in this paper.}
The nested structure of labelled blocks is defined by
\term{local contexts of depth $k$} ($L^k[\_]$). 
A local context of depth 0 ($L^0[\_]$) matches a stack of the form $v^*e^*$, 
and a local context of depth $k+1$ matches a local context of depth $k$ nested
inside a labelled block. For example, \cref{example-wasm-stack} has four
instructions within an $L^1$ context.

A WebAssembly program is organized as a collection of \term{modules} that import
and export code and data. An instantiated module with no unresolved imports
is called an \term{instance}. The global execution state of all instances
is called the \term{store}. We present \name as an extension to the WebAssembly
formal semantics, which includes the machinery needed to support multiple
instances. However, for the purpose of this paper, it is sufficient to
consider programs with just one instance.

The WebAssembly stack, nested control flow, the store, and instances are the elements of
WebAssembly that are relevant to \name. 
With these defined, WebAssembly has a  small-step semantics that
updates the stack, and possibly the store ($s$) and local variables ($v_l^*$) at each step
($s;v_l^*;e^* \hookrightarrow s';v_l'^*;e'^*$).
The semantics is congruent with local contexts:
if $s;v_l^*;e^* \hookrightarrow s';v_l'^*;e'^*$ 
then $s;v_l^*;L^k[e^*] \hookrightarrow s';v_l'^*;L^k[e'^*]$. Thus evaluation always occurs
in the innermost local context, unless no such evaluation is possible.
When an instruction does not read or write from the store, we omit the store for
brevity.
\Cref{example-wasm-red-seq} shows
the execution trace of our example.

\subsection{Design and Semantics of \name}
\label{sec:semantics-one-shot}

We first describe the new values and types of \name, then present necessary
changes to WebAssembly instances, and finally present the new reduction
rules of \name.

\begin{figure}
\footnotesize
\[
\begin{array}{r@{\;}c@{\;}ll}
\multicolumn{3}{l}{\textbf{Continuation IDs}} \\
\kid &::=& \mathbf{i64} \\[0.5em]
\multicolumn{3}{l}{\textbf{Instructions}} \\
e &::= &\cdots \\
  & \mid & \mathbf{control} \; h \\
  & \mid & \restore \\
  & \mid & \continuationcopy \\
  & \mid &\continuationdelete \\
  & \mid &\prompt{\textit{tf}}{e^*} \\[0.5em]
\multicolumn{3}{l}{\textbf{Full-Stack Contexts}} \\
\Lmax &::=& v^* [\_] e^* \mid v^* \; \mathbf{label}_n \{ e^* \} \; L^\mathrm{max} \; \op{end} \; e^* \\[0.5em]
\multicolumn{3}{l}{\textbf{Stores}} \\
s &::=& \{ \mathrm{inst} \; \textit{inst}^*,\cdots \} \\[0.5em]
\multicolumn{3}{l}{\textbf{Instances}} \\
\mathit{inst} &::=& \{ \mathrm{func} \; cl^*, \mathrm{glob} \; v^*, \mathrm{tab} \; i^?, \mathrm{mem} \; i^?, \highlight{\mathrm{pstack} \; \mathit{pstack}} \}  \\[0.5em]
\multicolumn{3}{l}{\textbf{Continuation Table Stacks}} \\
\mathit{pstack} & ::= & \mathit{pinst}^* \\[0.5em]
\multicolumn{3}{l}{\textbf{Continuation Tables}} \\
\mathit{pinst} & ::= & \{\mathrm{ctable} \; \left( \{\mathrm{locals} \; v^*, \mathrm{ctx} \; \Lmax, \mathrm{inst} \; i \} \mid \mathit{nil} \right)^*, \mathrm{root} \; (\kid \mid \mathit{nil}) \}
\end{array}
\]
\caption{Syntax of \name: we extend the WebAssembly runtime structure with a table of continuations.}
\label{wasmk-new-types}
\end{figure}
	
\paragraph{The continuation table and continuation IDs.}

In a language that supports first-class continuations, a continuation is a new
kind of value. First-class continuations are typically found in high-level
languages (e.g., Scheme or Racket) that also support first-class functions.
This allows functions that receive captured continuations (e.g., the argument
to \texttt{call/cc} in Scheme) to close over other variables in their environment.
However, this is not possible in WebAssembly, since it lacks first-class functions.
Moreover, it is not straightforward to safely add new kinds of values to
WebAssembly either. (The WebAssembly heap is
untyped and byte-addressable, so a program can make arbitrary changes to the
representation of any value stored on the heap.)
\name adds a \emph{continuation table}, which associates a
continuation with an integer-valued \emph{continuation ID} ($\kid$).
The \name runtime system manages the
table, and \name programs work with continuations indirectly by referring
to their ID. Since continuation IDs are standard integers, programs can use
existing load and store instructions to save continuation IDs
in linear memory. However, the new \name instructions have to dynamically ensure
that they receive valid continuation IDs. Our implementation uses 64-bit
integers to represent continuation IDs.

\name has delimited continuations, which are needed to safely
interoperate with host languages such as JavaScript~(\cref{sec:safety-properties}). 
Thus \name includes a $\op{prompt}$ instruction, and we modify instances to
track a stack of dynamic $\op{prompt}$ scopes. Formally, 
we extend instances as follows
(highlighted in  \cref{wasmk-new-types}).
Each instance ($\mathit{inst}$) contains a stack of dynamically 
nested $\op{prompt}$ contexts. Each prompt context ($\mathit{pinst}$) is a
record containing a 
continuation table and a continuation ID of the \term{root continuation}, 
defined to be the continuation associated with the stack which initialized the WebAssembly 
execution or most recent prompt. The continuation table consists of an array of entries, 
where each entry is either $\mathit{nil}$ or a captured continuation, with
all entries initialized to $\mathit{nil}$.\footnote{To control resource
utilization, WebAssembly implementations can define the maximum size of
various dynamic and static data structures, e.g., the number of stack
frames. Similarly, we impose an implementation-dependent bound on the number
of allocated continuations.}

A continuation saved in a continuation table is
a record that contains the values of local
variables (across all stack frames), and the entire evaluation context at the point of capture.
However, WebAssembly's local context ($L^k$) only capture a stack with $k$
nested blocks. Therefore, we define \term{full-stack contexts} ($\Lmax$) as
evaluation contexts that match a stack with arbitrary 
control block depth, and a continuation stores a full-stack context.

Finally, the root continuation ID in a continuation table is maintained such that
if the root continuation is currently executing, 
$\mathit{root}$ will be set to $\mathit{nil}$, otherwise $\mathit{root}$ will be set to the index $\kappa_R$ of
the root continuation in the continuation table. 

\begin{figure*}
\footnotesize
\boxed{s;\; v^*;\; e^* \rightsquigarrow_i s;\; v^*;\; e^*} \\
\vspace{1em}
[Cong]
\begin{prooftree}
	\hypo{s;\; v^*;\; e^* \hookrightarrow_i s';\; v'^*;\; e'^*}
	\infer1{s;\; v^*;\; L^k[e^*] \hookrightarrow_i s';\; v'^*;\; L^k[e'^*]}	
\end{prooftree}\hspace{2em}
[No-Ctrl]
\begin{prooftree}
	\hypo{s;\; v^*;\; e^* \hookrightarrow_i s';\; v'^*;\; e'^*}
	\infer1{s;\; v^*;\; e^* \rightsquigarrow_i s';\; v'^*;\; e'^*}	
\end{prooftree}\hspace{2em}
\vspace{1em}
\begin{alignat*}{3}
&\text{[Ctrl]} & s;\; v_l^*;\; \Lmax[(\mathbf{i64.const} \; v) \; (\mathbf{control} \; h)] &\rightsquigarrow_i s';\; \epsilon ;\; (\mathbf{i64.const} \; \kappa) \; (\mathbf{i64.const} \; v) \; (\mathbf{call} \; h) \; \mathbf{trap}&\quad &\text{if}~(s',\kappa) = \controltrans(s, i, v_l^*, \Lmax) \\
&\text{[Restore]} &s;\; v_l^*;\; \Lmax[(\mathbf{i64.const} \; \kappa) \; (\mathbf{i64.const} \; v) \; \mathbf{restore}] &\rightsquigarrow_i s';\; {v_l^*}' ;\; {\Lmax}' [(\mathbf{i64.const} \; v)]&&\text{if}~(s', v_l^*, {\Lmax}') = \restoretrans(s, i, \kappa) \\
&\text{[Restore-Err]} &s;\; v_l^*;\; \Lmax[(\mathbf{i64.const} \; \kappa) \; (\mathbf{i64.const} \; v) \; \mathbf{restore}] &\rightsquigarrow_i s;\; v_l^* ;\; \mathbf{trap} &&\quad\text{otherwise} \\
& & \\
&\text{[Copy]} &s;\; (\mathbf{i64.const} \; \kappa) \; \mathbf{continuation\_copy} &\hookrightarrow_i s';\; (\mathbf{i64.const} \; \kappa') &&\text{if}~(s', \kappa') = \copytrans(s, i, \kappa) \\
&\text{[Copy-Err]} &s;\; (\mathbf{i64.const} \; \kappa) \; \mathbf{continuation\_copy} &\hookrightarrow_i s;\; \mathbf{trap} &&\quad\text{otherwise}\\
&\text{[Delete]} &s;\; (\mathbf{i64.const} \; \kappa) \; \continuationdelete &\hookrightarrow_i s';\; \epsilon &&\text{if}~s' = \deletetrans(s, i, \kappa) \\
&\text{[Delete-Err]} &s;\; (\mathbf{i64.const} \; \kappa) \; \continuationdelete &\hookrightarrow_i s;\; \mathbf{trap} &&\quad\text{otherwise} \\
&\text{[Prompt]} &s;\; \prompt{\textit{tf}}{e^*} &\hookrightarrow_i s';\; \op{block} \; \textit{tf} \; e^* \; \op{end} \; \op{prompt\_end}  &&\text{if}~ s' = \prompttrans(s, i) \\
&\text{[Prompt-End]} &s;\; \op{prompt\_end} &\hookrightarrow_i s';\; \epsilon  &&\text{if}~ s' = \promptendtrans(s, i)
\end{alignat*}

\begin{align*}
\controltrans(s, i, v_l^*, \Lmax) &::= \begin{cases}
	(\fsetCont(\fsetRoot(s, i, \kappa), i, \kappa, \{ \mathrm{locals} = v_l^*, \mathrm{ctx} = \Lmax, \text{inst} = i \}), \; \kappa) & \quad\quad\quad\quad\quad\hspace{3.5pt} \text{if}~\fgetRoot(s, i) = \mathit{nil} \\
	(\fsetCont(s, i, \kappa, \{ \mathrm{locals} = v_l^*, \mathrm{ctx} = \Lmax, \text{inst} = i \}), \; \kappa) & \quad\quad\quad\quad\quad\hspace{3.5pt} \text{if}~\fgetRoot(s, i) \neq \mathit{nil}	
\end{cases} \\
&\text{where}~\kappa~\text{is fresh, i.e.,}\fgetCont(s, i, \kappa)=\mathit{nil} \\
\restoretrans(s, i, \kappa) &::= \begin{cases}
	(\fsetRoot(\fsetCont(s, i, \kappa, \mathit{nil}), i, \mathit{nil}), \; \fgetCont(s, i, \kappa)_\text{locals}, \; \fgetCont(s, i, \kappa)_\text{ctx}) & \quad\quad\hspace{0.8pt} \text{if}~\fgetRoot(s, i) = \kappa \\
	(\fsetCont(s, i, \kappa, \mathit{nil}), \; \fgetCont(s, i, \kappa)_\text{locals}, \; \fgetCont(s, i, \kappa)_\text{ctx}) & \quad\quad\hspace{0.8pt} \text{if}~\mathit{nil} \neq \fgetRoot(s, i) \neq \kappa
\end{cases} \\
\copytrans(s, i, \kappa) &::= (\fsetCont(s, i, \kappa', \fgetCont(s, i, \kappa)),\; \kappa') \quad\quad\quad\quad\quad\quad\quad\quad\quad\quad\quad\quad\quad\quad\quad\quad\quad\quad\quad\;\;\;\hspace{1.0pt} \text{if}~\fgetRoot(s, i)\neq \kappa \land \fgetCont(s, i, \kappa) \neq \mathit{nil} \\
&\text{where}~\kappa'~\text{is fresh, i.e.,}\fgetCont(s, i, \kappa')=\mathit{nil} \\
\deletetrans(s, i, \kappa) &::= \fsetCont(s, i, \kappa, \mathit{nil}) \quad\quad\quad\quad\quad\quad\quad\quad\quad\quad\quad\quad\quad\quad\quad\quad\quad\quad\quad\quad\quad\quad\quad\quad\quad\quad\quad\quad\hspace{-0.4pt} \text{if}~\fgetRoot(s, i)\neq \kappa \land \fgetCont(s, i, \kappa) \neq \mathit{nil} \\
\prompttrans(s, i) &::= s'~\text{where}~s' = s~\text{except}~s'_\text{inst}(i)_\text{pstack} \mapsto \text{push}(s_\text{inst}(i)_\text{pstack}, \{ \text{ctable} = \mathit{nil}^*, \text{root} = \mathit{nil}, \text{inst} = i \}) \\
\promptendtrans(s, i) &::= s'~\text{where}~s' = s~\text{except}~s'_\text{inst}(i)_\text{pstack} \mapsto \text{pop}(s_\text{inst}(i)_\text{pstack}) \quad\quad\quad\quad\quad\quad\quad\quad\quad\quad\quad\quad\quad\hspace{0.4pt} \text{if}~ \fgetRoot(s, i) = \mathit{nil}
\end{align*}

\begin{align*}
	\fgetRoot(s, i) &::= \text{top}(s_\text{inst}(i)_\text{pstack})_\text{root} \\
	\fgetCont(s, i, \kappa) &::= \text{top}(s_\text{inst}(i)_\text{pstack})_\text{ctable}(\kappa) \\
	\fsetRoot(s, i, \kappa_R^?) &::= s'~\text{where}~s' = s~\text{except}~\text{top}(s'_\text{inst}(i)_\text{pstack})_\text{root} \mapsto \kappa_R^? \\
	\fsetCont(s, i, \kappa, \gamma^?) &::= s'~\text{where}~s' = s~\text{except}~\text{top}(s'_\text{inst}(i)_\text{pstack})_\text{ctable}(\kappa) \mapsto \gamma^?
\end{align*}

\caption{Semantics of \name.}
\label{additional-semantic-rules}
\end{figure*}

\paragraph{New reduction rules.}

The semantics of WebAssembly define a reduction relation ($\hookrightarrow$)
that is congruent with local contexts ([Cong] in \cref{additional-semantic-rules}).
However, full congruence with local contexts does not hold in the presence of
first-class continuations. Therefore, \name introduces a new reduction relation
($\rightsquigarrow$) for programs that contain $\control{h}$ and $\restore$
instructions (\cref{additional-semantic-rules}). The extended semantics refer
to the original WebAssembly reduction relation ($\hookrightarrow$), using the
[Cong] rule, but there is no equivalent rule for $\rightsquigarrow$.
If there
is a reduction which involves no use of $\control{h}$ or $\restore$, then it is also a 
valid reduction which \emph{might} make use of $\control{h}$ or $\restore$, as given
in the [No-Ctrl] rule.

The $\control{h}$ instruction receives a single argument ($v$) and calls
the function $h$, passing it a new continuation ID ($\kappa$) and the argument
$v$. The continuation ID is bound to the current continuation ($\Lmax$) and local variables ($v_l^*$), and the
call to $h$ is followed by a $\op{trap}$: i.e., it is a runtime error to return
normally from $h$. For simplicity, $\control{h}$ makes a direct call to a
function $h$. However, when an indirect call is necessary, it is possible to
use $v$ to pass the index of a function to $h$.
  
The $\mathbf{restore}$ instruction receives a continuation ID ($\kappa$) and
a restore value ($v$). The instruction dynamically checks that $\kappa$ is a
valid continuation ID. If $\kappa$ is valid, it restores the local variables
(${v_l^*}'$) and the stack (${\Lmax}'$) that is associated with $\kappa$, and
returns $v$ to the stack. The $\mathbf{restore}$ instruction also marks the
continuation ID ($\kappa$) as $\mathit{nil}$ in the continuation table, which allows it
to be reused by subsequent calls to $\op[h]{control}$. Finally, when restoring
the root continuation, $\op{restore}$ sets the root ID back to $\mathit{nil}$, and
leaves it untouched otherwise.
Note that $\restore$ is abortive rather than functional, in
the sense that $\restore$ aborts the current continuation and instructions
following $\restore$ will never be executed.
It is a runtime error to call $\restore$ on a continuation ID ($\kappa)$ that
is un-allocated, or to invoke $\restore$ within the root continuation. In
either case, a $\op{trap}$ occurs.

We need the $\continuationcopy$ instruction to create a copy of a saved
continuation, so that a program can restore a continuation several times if
needed. This instruction assigns a new continuation ID to the copy. 
A $\op{trap}$ occurs if the provided
continuation ID is mapped to $\mathit{nil}$. The $\continuationdelete$
instruction deallocates an continuation without restoring it, and may be needed
to avoid memory leaks in certain applications.

In the presence of first-class continuations,
a function $f$ may now never return to the call site or may
return multiple times. Motivated by a need for safe FFI, the goal of a 
$\prompt{\textit{tf}}{e^*}$ instruction
\footnote{
	$\mathit{tf}$ is a type annotation of the body ($e^*$) of the prompt, 
	and is not important to understand the semantics.
}
is to evaluate the body $e^*$ such that $e^*$ is guaranteed to finish evaluation exactly once 
(or trap/diverge), and trap otherwise. 
Note that this is similar to Felleisen's prompt~\cite{felleisen:1988:the},
but in cases where Felleisen's prompt alters the control flow, \name's $\op{prompt}$ traps.
This design is due to the fact that our $\restore$ operator is abortive 
rather than functional.
Evaluation of $\prompt{\textit{tf}}{e^*}$ involves first pushing a
prompt context onto the prompt stack with a blank continuation table and the 
$\textit{root}$ ID set to $\textit{nil}$, then executing
$e^*$ inside a scoped block, and finally executing the administrative non-user
accessible instruction $\promptend$. Note that if $e^*$ were to contain branches 
to labels outside of the $\op{prompt}$, the execution of $\promptend$ could be 
skipped. The validation rules discussed below outlaw such branches. 
The safety properties of $\op{prompt}$ during FFI is discussed in \cref{sec:safety-properties}.
Evaluating a $\promptend$ instruction pops and discards the top 
prompt context from the prompt stack.

\begin{figure}[t]
\footnotesize
\[
C ::= \{ \dots, \textrm{label} \; ((t^*)^*)^*, \text{pstack} \{ \text{ctable} (t^* \mid \mathit{nil})^*, \text{root} (\kappa_R \mid \mathit{nil}) \}^* \}\vspace{1em}
\]
\begin{prooftree}
	\hypo{C_{\mathrm{func}}(h) = \mathrm{i64} \; \mathrm{i64} \to \epsilon}
	\infer1{C \vdash \control{h} : \mathrm{i64} \to \mathrm{i64}}	
\end{prooftree}\\\vspace{2em}
\begin{prooftree}
	\infer0{C \vdash \restore : t_1^*\; \mathrm{i64} \; \mathrm{i64} \to t_2^*}	
\end{prooftree}\\\vspace{2em}
\begin{prooftree}
	\infer0{C \vdash \continuationcopy : \mathrm{i64} \to \mathrm{i64}}	
\end{prooftree}\\\vspace{2em}
\begin{prooftree}
	\infer0{C \vdash \continuationdelete : \mathrm{i64} \to \epsilon}	
\end{prooftree}\\\vspace{2em}
\begin{prooftree}
	\hypo{\textit{tf} = t_1^n \to t_2^m}
	\hypo{C \{\textrm{label} = C_\textrm{label}; ((t_2^m)), \; \textrm{return} = \epsilon \} \vdash e^* : \textit{tf}}
	\infer2{C \vdash \prompt{\textit{tf}}{e^*} : \textit{tf}}
\end{prooftree}
\caption{Type Checking of \name.}
\label{additional-verification-rules}
\end{figure}
\begin{figure}
\footnotesize
[Root]\quad\begin{prooftree}
	\hypo{\vdash_i s ;~v^*;~e^* : t^*}
	\hypo{\vdash s : S}
	\hypo{S_\text{inst}(i)_\text{roots} = \textit{nil}^*}
	\infer3{\vdash_i^k s ;~v^*;~e^* : t^*}
\end{prooftree}\\\vspace{2em}
[Non-Root]\quad\begin{prooftree}
	\hypo{\vdash s : S}
	\infer[no rule]1{\vdash_i s ;~v^*;~e^* : t^*}
	\hypo{
		p_R = \max \{ p \mid S_\text{inst}(i)_\text{pstack}(p)_\text{root} \neq \mathit{nil} \}
	}
	\infer[no rule]1{\kappa_R = S_\text{inst}(i)_\text{pstack}(p_R)_\text{root}}
	\infer[separation=1em]2{\vdash_i^k s ;~v^*;~e^* : S_\text{inst}(i)_\text{pstack}(p_R)_\text{ctable}(\kappa_R)}
\end{prooftree}
\caption{Typing delimited instructions.}
\label{fig:store-typing}
\end{figure}


\subsection{Validation}
\label{sec:validation}
Validation (type checking) is accomplished in WebAssembly by
assigning each instruction a type describing the values it pops from the stack
and the values it pushes onto the stack. For example, the type of an
add instruction ($\op{i32.add}$) is 
$\mathrm{i64} \; \mathrm{i64} \to \mathrm{i64}$. In addition, the context ($C$)
stores information during the type checking algorithm, such as the types of
functions.

\cref{additional-verification-rules} shows the type checking rules for \name.
The type checking of \name fits easily into the existing type checking framework
of WebAssembly, since we check dynamically that continuation IDs are valid, 
similar to the type checking of indirect function calls.

The type checking of $\restore$, $\continuationcopy$, and $\continuationdelete$
instructions is straightforward as they are all typed independent of
the context ($C$). In particular, these instructions do not statically
type check validity of continuation IDs, 
beyond being the correct type ($\mathrm{i64}$), 
since the semantics in \cref{additional-semantic-rules} check continuation ID
validity at runtime. The type checking of a $\control{h}$ instruction does 
involve checking a side condition in the context: 
in order to type check $\control{h}$, the handler function ($h$) is looked up
in the context ($C$), and checked to have the correct type of a control handler 
function (receives two $\mathrm{i64}$ arguments and returns nothing).

Type checking the $\op{prompt}$ instruction is the most interesting case.
Semantically, $\prompt{\textit{tf}}{e^*}$ must 1)~prepare the prompt environment, 
2)~execute $e^*$, and 3)~teardown the prompt environment 
(i.e., execute the $\promptend$ administrative instruction). However, consider that
$e^*$ may contain branch instructions jumping to labels lexically outside of
the $\op{prompt}$, which would then incorrectly be able to jump beyond 
tearing-down of the prompt environment. To remedy this, we use the type checker to 
outlaw branching instructions which jump beyond the scope of the 
$\op{prompt}$, though still allow branches within the $\op{prompt}$.

We extend type-checking contexts ($C$) to store
a stack of stacks of labels, as shown in the top of 
\cref{additional-verification-rules} ($\textrm{label} ((t^*)^*)^*$).
Implicitly, we define the notation of context label extension 
$C, \textrm{label} (t^*)$ used in previous WebAssembly type checking rules to mean
that the label $(t^*)$ is pushed onto the \emph{top-most} stack in $C$ 
(or in a new stack if none exist), and 
likewise the notation $C_{\textrm{label}(i)}$ we define to mean indexing by $i$ into 
the \emph{top-most} stack in $C$. These implicit re-definitions allow all the other
WebAssembly type checking rules to remain untouched. With this machinery in place,
the type checking rule for $\op{prompt}$ can be given, which closely mirrors the
type checking rule of $\op{block}$, except that an entire new label stack is 
pushed into the context and the return label is invalidated.

An alternative approach could be to modify the
semantics to force the $\promptend$ instruction to be run even when branching past
it. However, this would require significant changes to how branch instructions
are specified in the WebAssembly semantics, and would significantly impact
code generation.

\subsection{Safety Properties of \name}
\label{sec:safety-properties}

We first prove the safety of \name, building on the safety of WebAssembly.
We then consider safe interoperation with a host language.


\paragraph{Safety of standalone \name.}

WebAssembly is equipped with a syntactic type soundness
theorem~\cite{haas:2017:bringing, watt:2018:mechanising, wasm:coq}, which
we build on.

WebAssembly's instruction typing relation ($\vdash_i e^*;t^*$) calculates a
sequence of types ($t^*$), which specify the types of the values that are left
on the stack by the instructions ($e^*$). These types are preserved by each
step of evaluation ($\hookrightarrow$). However, if a step 
captures or restores a continuation ($\rightsquigarrow$), the type of the
current instruction sequence may change.

To address this, we introduce a new typing relation ($\vdash_i^k$) which
extracts the type of the unique stack nested most deeply in prompts which has
been invoked through a chain of root stacks (\cref{fig:store-typing}).
We call this stack the \emph{primary root stack}. 
There are two cases to this relation: 1)~when the current instruction sequence
 is the primary root stack, we return its type ([Root]), and 2)~
if not, we extract the type of the saved primary root stack from the store
([Non-Root]).
Using this typing relation, we prove progress and preservation for \name.

\begin{theorem}[$\rightsquigarrow$ Preservation]
	If $\vdash_i^k s;~v^*;~e^* : t^*$ and $s;~v^*;~e^* \rightsquigarrow_i s';~v'^*;~e'^*$, then $\vdash_i^k s';~v'^*;~e'^* : t^*$.
\end{theorem}

\begin{theorem}[$\rightsquigarrow$ Progress]
	If $\vdash_i^k s; v^*; e^* : t^*$, then either $e^* = v'^*$ or $e^* = \textbf{trap}$ or $s; v^*; e^* \rightsquigarrow_i s'; v'^*; e'^*$.
\end{theorem}

\noindent\ifappendix
The proofs of both theorems are available in the appendix.
\else
The proofs of both theorems are available at \url{https://wasmk.github.io}.
\fi

\paragraph{Safe interoperation.}

A WebAssembly runtime environment is typically embedded in a host language, and
offers an API that allows function calls from either language to the other.
For example, Wasmtime supports interoperability with Rust, and web browsers
support interoperability with JavaScript. Neither Rust nor JavaScript support
continuations, and require foreign function calls to return exactly once.
The $\op{prompt}$ operator allows us to enforce this dynamically.
\name automatically inserts a $\op{prompt}$ block around a foreign call into \name.
This design is
similar to Scheme / Racket, but differs in two regards. First, Scheme / Racket allow FFI to
be unsafe as they do not forcibly wrap every FFI call in a \texttt{prompt}, while \name
prioritizes safety over some flexibility. Second, Scheme / Racket will not abort the program
upon control flow which violates the exactly-once semantics of \texttt{prompt}, but will
instead alter the control flow~\cite{felleisen:1988:the}. In keeping with using $\op{prompt}$ strictly to enforce FFI
safety, \name considers it a programmer error to attempt to violate such safety.

\section{Leveraging \name in Existing Compilers}
\label{sec:in-practice}

Since \name does not alter the semantics of existing WebAssembly instructions,
it ought to be easy to use \name to implement continuations in an existing
compiler. However, today's compilers use code generation techniques that
require a little extra care.

For example, consider Emscripten, which compiles C to WebAssembly. A typical C
compiler would allocate local variables on the machine stack, and Emscripten is
no exception. However, whereas a C program can obtain a pointer to a local,
stack-allocated variable---a common operation in C programs---it is not
possible to do so in WebAssembly. The WebAssembly stack is not stored in linear
memory, and programs can only obtain pointers to values in linear memory.
Therefore, to support these programs, Emscripten allocates local variables on
the WebAssembly stack when possible, but uses linear memory when necessary.
Emscripten generates code that reserves a block of memory to store the
heap-allocated portion of the stack, and uses global variables that emulate
stack and frame pointers.

\cref{sec:c-k} presented \cname, which extends Emscripten with
continuations. Our extension adds new library functions that each correspond to
a \name instruction, which manage saving and restoring the WebAssembly stack.
However, we need to ensure that these operations correctly save and restore
the heap-allocated portion of the stack (\name cannot do this automatically,
since it is source-language neutral). Therefore, we insert code at the call
site for each \cname operation to manipulate Emscripten's
global stack pointer values. For example, at a call site of \texttt{control},
we insert code that saves the current heap-allocated portion of
the stack and the value of the stack pointer into a table. Similarly, at
a call site of \texttt{restore},  we insert code that restores
the heap-allocated portion of the stack and stack pointer
from the table.

This problem is not unique to Emscripten. For example, the Go compiler's
WebAssembly backend creates a copy of the WebAssembly stack in linear memory to
support garbage collection and Goroutines. We speculate that \name would
allow the Go compiler to store non-pointer variables on the
WebAssembly stack, which may improve
the performance of numeric code. However, GC roots would still have to be
stored in linear memory.

\section{Implementation}
\label{sec:jit}

We implement \name as an extension to \emph{Wasmtime}, which is a standalone,
JIT-based runtime system for WebAssembly.\footnote{Our implementation is available at \url{https://wasmk.github.io}.} Wasmtime is written in Rust and
primarily developed by Mozilla.
Wasmtime, and other WebAssembly JITs, use the native machine stack to store
both values and return addresses. Wasmtime also performs register allocation
to avoid using the stack when possible. Therefore, our \name implementation
has to manage the native stack and registers, and take care to follow
the calling convention that Wasmtime employs.

\paragraph{Capture and restore.}

The implementation of $(\mathbf{control} \; h)$ involves several steps.
1)~It uses a free list to 
allocate an unused continuation ID. 2)~It associates this continuation ID with
a new continuation object, which holds the values of machine registers
that are not caller-saved, which includes the stack and instruction
pointers. 3)~It allocates a new block of memory to hold subsequent stack
frames, and sets the stack pointer to point to this block of memory.
4)~It jumps to the WebAssembly function $h$, which receives the
new continuation ID. To further improve performance, we preallocate a pool
of memory to hold new stacks.

The implementation of $\mathbf{restore}$ is straightforward, since its
principal task is to restore the registers saved by $\mathbf{control}$ in
the continuation object. To ensure safety, we 1)~ensure that the continuation
ID is associated with a valid continuation object, 2)~delete the continuation
object so that it cannot be restored again, and 
3)~reclaim the memory used by the current stack.

\paragraph{Copying continuations.}

To copy a continuation, we allocate a new continuation ID, and duplicate a
continuation object, but have to carefully tackle all pointers within the
continuation object. The continuation object stores an instruction pointer,
which can be freely copied. However, we have to update the saved stack pointer
to point to the duplicate copy of the saved stack. This is sufficient for
Wasmtime, but other implementations may require extra work. For example, if an
implementation stores pointers into the stack in registers or on the stack
itself, they must be updated to point to the copy.




\section{Evaluation}
\label{sec:Evaluation}

\begin{figure}
\centering
\includegraphics[width=\columnwidth]{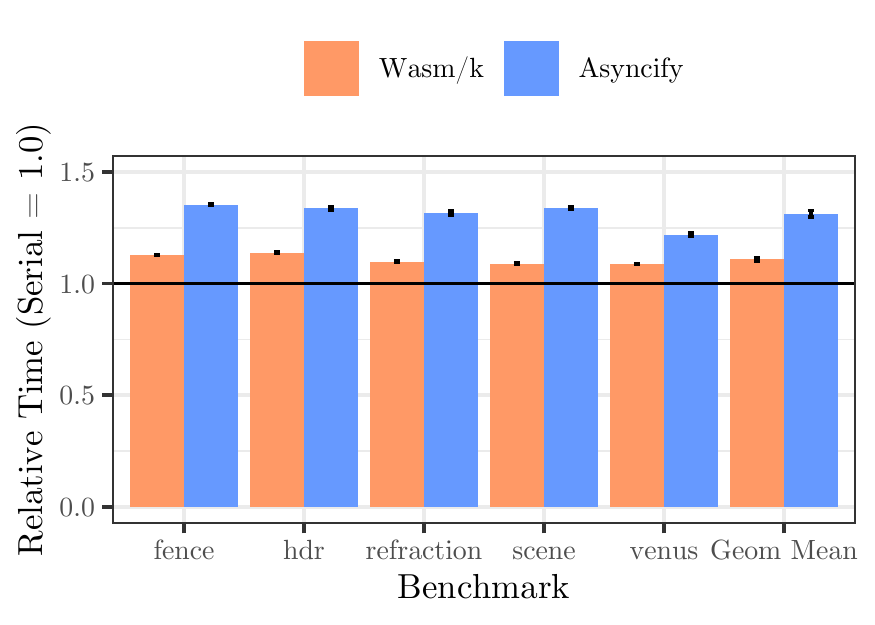}
\caption[Caption for C-Ray]{Performance of green threads implemented using \name and Asyncify in a ray tracing application \footnotemark. All experiments were performed on a 64 bit 3.3GHz 4-Core CPU on Ubuntu. Error bars show the 95\% confidence interval of the overhead, over six trials.}
\label{fig:plot-ray1}
\end{figure}
\footnotetext{Commit \texttt{21124ee} of the C-Ray fork available at \url{https://wasmk.github.io} was used in this experiment.}

In this section we compare \name to the natural alternative: which is to
implement continuations using a whole program transformation that doesn't require
any change to WebAssembly.

\term{Asyncify}~\cite{wasm:asyncify} is a tool that simulates non-blocking I/O
in WebAssembly. It extends WebAssembly with control operators that are similar
to one-shot continuations, and outputs standard WebAssembly that simulates
control flow. We use both \name and Asyncify to implement a green threading
library (\cref{sec:using-wasmk}), which allows us to directly compare the
performance of threaded programs.

As a benchmark, we use C-Ray~\cite{wasm:c-ray}, which is a ray tracer
implemented in approximately 9,500 lines of C. Ray tracers are
compute-intensive, and take a long time to render the final, full-quality
image. However, because they can compute lower quality rendering approximations
incrementally, it should be possible to display incremental rendering results,
to appear more responsive to the user. C-Ray performs rendering computations
on background threads (using
\texttt{pthreads}), while the main thread periodically displays the current
scene. We ported C-Ray to use our green threading library, and inserted
thread yields in the main rendering loop, which yielded
approximately once every 25 ms.

\paragraph{Code size.}

Code size is particularly important for web browsers, which download code on
demand, and a key factor of WebAssembly's design is that it has a compact
binary file format. The size of the C-Ray WebAssembly program is
$1.3\times$ larger with Asyncify than it is with \name.

\paragraph{Performance.}

\cref{fig:plot-ray1} shows the time needed to complete ray tracing on five different
visual scenes using \name and Asyncify, with the geometric mean over all scenes shown
in the last two columns on the right.
As a baseline, we use C-Ray running in
WebAssembly with no threading. Note that the baseline has limited utility,
since it cannot show intermediate results. However, it does illustrate the
overhead that both \name and Asyncify introduce. The mean running time of \name
is $1.1\times$ the running time
without threads. In contrast, the mean running time of Asyncify is
$1.3\times$ the running time
without threads. Asyncify is slower, since it introduces several loads, stores,
and branches to the compiled code. The smaller slowdown that \name introduces is the cost of
checking whether it is time to switch threads.


\section{Related Work}
\label{sec:Related Work}

\paragraph{WebAssembly.}

\name extends the formal semantics of WebAssembly
1.0~\cite{haas:2017:bringing}. There are several proposed extensions to
WebAssembly 1.0, not all of which have been implemented in production web
browsers. 
The threading proposal~\cite{wasm:threads} extends WebAssembly
with support for atomic memory operations and synchronization
primitives, but leaves the API for thread creation up to
each WebAssembly runtime implementation. Thus far, the \texttt{pthread}
API has been supported in some browsers. Watt et al. build on the threading
proposal by formalizing a semantics and memory model for concurrent
WebAssembly~\cite{watt:2019:weakening}. This work on robust support
for concurrency via physical threads is an important step for WebAssembly,
and is orthogonal and complimentary to \name: 
both aspects are needed for efficient implementations
of goroutines which can utilize all CPU cores.

Another proposal extends WebAssembly with support 
for exception handling~\cite{wasm:exceptions}, which is a
form of limited stack manipulation. An interesting question of semantics not
addressed in this work is how \name would interact with exception handling.
In this direction, there is prior work on supporting both
delimited continuations and exception handling~\cite{flatt:2007:adding}.

An alternative to supporting continuations natively is to implement them by
source-to-source transformation~\cite{baxter:2018:putting,pettyjohn:cm}.
Asyncify~\cite{wasm:asyncify} does so for WebAssembly, and the Go compiler uses
a similar approach to support Goroutines. Our evaluation
(\cref{sec:Evaluation}) shows that \name is significantly faster than
source-to-source transformation, and produces smaller programs.

A recent discussion sketched an alternative design for WebAssembly
continuations~\cite{wasm:effects} that is based on extending exception
handlers with general effect handlers. Our design is orthogonal to
exception handling and makes fewer changes to the WebAssembly 1.0
type system. To the best of our knowledge, this alternative design has not
been implemented at this time.

Finally, there exist related strategies of program execution
control. Existing interpreters or virtual machines which feature execution control
mechanisms can be compiled to WebAssembly, such as the Lua VM 
(implemented in C)~\cite{wasm:lua-vm-wasm} which features coroutines. 
This is certainly a viable and straightforward strategy to allow 
stack-manipulating code to run in a WebAssembly environment, 
but may not be able to achieve performance comparable to compiling to WebAssembly.
In addition, debuggers can be seen as a form of execution control, 
as code can be paused and resumed,
but unlike with first-class continuations, the program control is not internally observable.
Debugger support for WebAssembly has recently been explored in the context of
microcontrollers~\cite{wasm:WARDuino}.

\paragraph{Continuations.}

We adapt Sitaram and Felleisen's \term{control}
operator~\cite{sitaram:1990:control} for WebAssembly. Our design accounts for
the fact that WebAssembly has neither first-class functions, nor garbage
collection: programs must explicitly delete unused continuations, and our
new control operators take additional arguments that are not necessary in
languages that support closures. We rely on control delimiters to ensure that
WebAssembly programs always safely interoperate with host languages that do not
support continuations, such as JavaScript. However, it should be possible to
adapt other control operators as
well~\cite{felleisen:1988:the,danvy:1990:abstracting}.

A goal of \name is to show that delimited continuations can be implemented
efficiently in a modern WebAssembly JIT. Our implementation uses a contiguous
stack, since it does not require global changes to code generation. However,
there are a variety of other implementation strategies with different
tradeoffs~\cite{farvardin:2020:from}.

\section{Conclusion}
\label{sec:Conclusion}
We have presented \name, an extension to WebAssembly that adds support
for delimited one-shot continuations with explicit copying. We have prototyped
all phases of \name, with examples in C/C++, code generation from C/C++ to \name,
formal semantics of \name, and an efficient implementation of \name in an existing JIT.
We hope that \name is a step toward helping WebAssembly be an effective compilation target
for a large variety of high-level languages.

\ifappendix
\else
\balance
\fi

\ifanonymous
\else
\section*{Acknowledgements}

This work was partially supported by the National Science Foundation under grants CCF-2007066, 
CCF-1453474, 
and CCF-1564162. 

\fi

\bibliographystyle{ACM-Reference-Format}
\bibliography{bib/venues-short,bib/main}

\ifappendix

\newpage

\appendix
\onecolumn

\section{Appendix}
\label{appendix:Technical Appendix}

\subsection{Continuation Table Typing Relation}
\cref{fig:store-typing} presented the core rules for the continuation
table typing relation $\vdash_i^k$. 
\cref{fig:store-typing} relies on an extension of
the WebAssembly store typing, in which a store $s$ is given a type $S$,
now extended with the $\text{pstack}$ field as shown in 
\cref{additional-verification-rules}. Modifications to existing WebAssembly store 
typing rules and additional rules are required to compute
the $\text{pstack}$ field of $S$. 
These modifications and new rules are given in 
\cref{additional-store-typing-appendix}, 
where modified WebAssembly rules are marked with a $(\dagger)$.

\begin{figure}[h]
\footnotesize
\begin{prooftree}
	\hypo{(S \vdash \mathit{cl} : \mathit{tf})^*}
	\hypo{(\vdash v : t)^*}
	\hypo{(S_\text{tab}(i) = n)^?}
	\hypo{(S_\text{mem}(j) = m)^?}
	\hypo{(S;~\epsilon \vdash_\mathit{ci} v_l^*;~\Lmax[\op[0]{i64.const}] : \mathit{tk}^*)^{?**}}
	\hypo{(\text{ctable}(\kappa_R) \neq \mathit{nil})^{?*}}
	\infer6{
		S \vdash \{ \text{func}~\mathit{cl}^*, \text{glob}~v^*, \text{tab}~i^?, \text{mem}~j^?, \text{pstack}~\{ \text{ctable}~(\{ \text{locals}~v_l^*, \text{ctx}~\Lmax, \text{inst}~\mathit{ci} \} \mid \mathit{nil})^*, \text{root} (\kappa_R \mid \mathit{nil}) \}^* \}	}
	\infer[no rule]1{
		: \{ \text{func}~\mathit{tf}^*, \text{global}~(\text{mut}^?~t)^*, \text{table}~n^?, \text{memory}~m^?, \text{pstack} \{ \text{ctable}(\mathit{tk}^* \mid \mathit{nil})^*, \text{root}(\kappa_R \mid \mathit{nil})\}^*\}
	}
\end{prooftree}
\quad $(\dagger)$\\\vspace{2em}
[Prompt-End]\quad\begin{prooftree}
	\hypo{S_\text{inst}(i)_\text{pstack}(0)_\text{root} = \mathit{nil}}
	\infer1{S;~C \vdash_i \promptend : \epsilon \to \epsilon}
\end{prooftree}
\caption{Additional Store Typing Rules.}
\label{additional-store-typing-appendix}
\end{figure}

In addition, we define the shorthand notation 
$S_\text{inst}(i)_\text{roots}$ used in \cref{fig:store-typing}
to be the vector $S_\text{inst}(i)_\text{pstack}(\cdot)_\text{root}$.

\subsection{Proofs of Type Safety Properties}

{
\setlength{\parindent}{0pt}
\newlength{\enumindent}
\setlength{\enumindent}{15pt}

\renewlist{enumerate}{enumerate}{9}

\begin{lemma}[Context Substitution]
\label{lemma:substitution}
If
\begin{enumerate}
	\item[1.] $\vdash_i s;~v^*;~L^k[e^*] : t^*$, and
	\item[2.] $\vdash_i s;~v^*;~e^* : t_e^*$, and
	\item[3.] $\vdash_i s;~v^*;~e'^* : t_e^*$,
\end{enumerate}
then $\vdash_i s;~v^*;~L^k[e'^*] : t^*$
\end{lemma}

\begin{proof}
This is a direct consequence of the transitive rule for WebAssembly type checking.
\end{proof}

\begin{theorem}[$\rightsquigarrow$ Preservation]
	If $\vdash_i^k s;~v^*;~e^* : t^*$ and $s;~v^*;~e^* \rightsquigarrow_i s';~v'^*;~e'^*$, then $\vdash_i^k s';~v'^*;~e'^* : t^*$.
\end{theorem}

\begin{proof}
Suppose that:
\begin{enumerate}
	\item[H1)] $\vdash_i^k s;~v^*;~e^* : t^*$ and
	\item[H2)] $s;~v^*;~e^* \rightsquigarrow_i s';~v'^*;~e'^*$
\end{enumerate}
We want to show that $\vdash_i^k s';~v'^*;~e'^* : t^*$.

\noindent By H1 and the typing rules [Root] and [Non-Root] for $\vdash_i^k$, we know:
\begin{enumerate}
	\item[H3)] $\vdash s : S$
\end{enumerate}
Based on the [Root] and [Non-Root] typing rules for H1, there are two cases:
\begin{enumerate}[label=\underline{Case C\arabic*}, align=left, leftmargin=1\enumindent]
	\item $S_\text{inst}(i)_\text{roots} = \mathit{nil}^*$. In this case we also know:
	\begin{enumerate}[align=left, leftmargin=1\enumindent]
		\item[H4)] $\vdash_i s;~v^*;~e^* : t^*$
	\end{enumerate}
	By H2 there are 10 subcases to consider:
	\begin{enumerate}[label*=\underline{.\arabic*}, align=left, leftmargin=1\enumindent]
		\item $s; v^*; \Lmax[\op[v]{i64.const} \control{h}] \rightsquigarrow_i s'; \epsilon ; \op[\kappa]{i64.const} \; (\mathbf{i64.const} \; v) \; (\mathbf{call} \; h) \; \mathbf{trap} \land (s',\kappa) = \controltrans(s, i, v^*, \Lmax)$
		
			By Case C1 and C1.1:
			\begin{enumerate}[align=left, leftmargin=1\enumindent]
				\item[H5)] $\controltrans(s, i, v^*, \Lmax) = (s', \kappa) = (\fsetCont(\fsetRoot(s, i, \kappa), i, \kappa, \{ \mathrm{locals} = v_l^*, \mathrm{ctx} = \Lmax \}), \kappa)$
				\item[H6)] $\fgetCont(s, i, \kappa) = \mathit{nil}$
			\end{enumerate}
			
			By H3 and H5 we can type check $s'$:
			\begin{enumerate}[align=left, leftmargin=1\enumindent]
				\item[H7)] $\vdash s' : S'$
				\item[H8)] $S'_\text{inst}(i)_\text{pstack}(0)_\text{root} = \kappa$
				\item[H9)] $\forall p \geq 1, S'_\text{inst}(i)_\text{pstack}(p)_\text{root} = \mathit{nil}$
			\end{enumerate}
			
			We can compute the type of the stored stack in $s'$. By \cref{lemma:substitution}, H4 and H5:
			\begin{enumerate}[align=left, leftmargin=1\enumindent]
				\item[H10)] $S'_\text{inst}(i)_\text{pstack}(0)_\text{ctable}(\kappa) = t^*$
			\end{enumerate}
			
			The new stack can be independently type checked. 
			By the $\op{control}$ type checking rule, 
			$S'_\text{inst}(i)_\text{func}(h) = \text{i64}~\text{i64} \to \epsilon$. Thus,
			\begin{enumerate}[align=left, leftmargin=1\enumindent]
				\item[H11)] $\vdash_i s';~\epsilon;~\op[\kappa]{i64.const}\op[v]{i64.const}\op[h]{call}\op{trap} : \epsilon$
			\end{enumerate}
			
			By H7, H8, H9, H10, H11 and the [Non-Root] rule:
			\begin{enumerate}[align=left, leftmargin=1\enumindent]
				\item[H12)] $\vdash_i^k s';~\epsilon;~\op[\kappa]{i64.const}\op[v]{i64.const}\op[h]{call}\op{trap} : t^*$
			\end{enumerate}
			Case C1.1 is complete.

		\item $s; v^*; \Lmax[\op[\kappa]{i64.const} \op[v]{i64.const} \restore] \rightsquigarrow_i s';\; v'^* ;\; {\Lmax}' [(\mathbf{i64.const} \; v)] \land (s', v^*, {\Lmax}') = \restoretrans(s, i, \kappa)$
		
			By C1.2 $\fgetRoot(s, i) \neq \mathit{nil}$. However this is a contradiction with case C1 which implies $\fgetRoot(s, i) = \mathit{nil}$. 
			
			Therefore, Case C1.2 is impossible.
		
		\item $s; v^*; \Lmax[\op[\kappa]{i64.const} \op[v]{i64.const} \restore] \rightsquigarrow_i s;\; v^* ;\; \op{trap} \land \restoretrans(s, i, \kappa)~\text{undefined}$
		
			By the type checking of $\op{trap}$,
			\begin{enumerate}[align=left, leftmargin=1\enumindent]
				\item[H13)] $\vdash_i s;\; v^* ;\; \op{trap} : t^*$
			\end{enumerate}
			By Case C1, H3, H13, and the [Root] rule, we conclude:
			\begin{enumerate}[align=left, leftmargin=1\enumindent]
				\item[H14)] $\vdash_i^k s;\; v^* ;\; \op{trap} : t^*$
			\end{enumerate}
			Case C1.3 is complete.
		
		\item $s;v^*; L^k[\op[\kappa]{i64.const} \continuationcopy] \hookrightarrow_i s';v^*; L^k[\op[\kappa']{i64.const}] \land (s', \kappa') = \copytrans(s, i, \kappa)$
		
			By Case C1.4:
			\begin{enumerate}[align=left, leftmargin=1\enumindent]
				\item[H15)] $\copytrans(s, i, \kappa) = (s', \kappa') = (\fsetCont(s, i, \kappa', \fgetCont(s, i, \kappa)), \kappa')$
				\item[H16)] $\fgetCont(s, i, \kappa') = \mathit{nil}$.
			\end{enumerate}
			
			By H4, Case C1.4, and \cref{lemma:substitution},
			\begin{enumerate}[align=left, leftmargin=1\enumindent]
				\item[H17)] $\vdash_i s';~v^*;~L^k[\op[\kappa']{i64.const}] : t^*$
			\end{enumerate}
			
			By H3, there exists a type for the stack that is being copied:
			\begin{enumerate}[align=left, leftmargin=1\enumindent]
				\item[H18)] $\exists t'^*,\; S_\text{inst}(i)_\text{pstack}(0)_\text{ctable}(\kappa) = t'^*$
			\end{enumerate}
			
			By Case C1, H15, H18, and the store typing relation,
			\begin{enumerate}[align=left, leftmargin=1\enumindent]
				\item[H19)] $\vdash s' : S'$
				\item[H20)] $S'_\text{inst}(i)_\text{roots} = \mathit{nil}^*$
			\end{enumerate}
			
			By H17, H19, H20 and the [Root] rule:		
			\begin{enumerate}[align=left, leftmargin=1\enumindent]
				\item[H21)] $\vdash_i^k s';v^*; L^k[\op[\kappa']{i64.const}] : t^*$
			\end{enumerate}
			Case C1.4 is complete.

		\item $s;v^*; L^k[\op[\kappa]{i64.const} \continuationcopy] \hookrightarrow_i s;v^*; L^k[\op{trap}] \land \copytrans(s, i, \kappa)~\text{undefined}$
		
			By the type checking of $\op{trap}$, H4, and \cref{lemma:substitution},
			\begin{enumerate}[align=left, leftmargin=1\enumindent]
				\item[H22)] $\vdash_i s;\; v^* ;\; L^k[\op{trap}] : t^*$
			\end{enumerate}
			By Case C1, H3, H22, and the [Root] rule, we conclude:
			\begin{enumerate}[align=left, leftmargin=1\enumindent]
				\item[H23)] $\vdash_i^k s;\; v^* ;\; L^k[\op{trap}] : t^*$
			\end{enumerate}
			Case C1.5 is complete.
		
		\item $s;v^*; L^k[\op[\kappa]{i64.const} \continuationdelete] \hookrightarrow_i s';v^*; L^k[\epsilon] \land s' = \deletetrans(s, i, \kappa)$
		
			By Case C1.6:
			\begin{enumerate}[align=left, leftmargin=1\enumindent]
				\item[H24)] $\deletetrans(s, i, \kappa) = s' = \fsetCont(s, i, \kappa, \mathit{nil})$
			\end{enumerate}
			
			By H4, Case C1.6, and \cref{lemma:substitution},
			\begin{enumerate}[align=left, leftmargin=1\enumindent]
				\item[H25)] $\vdash_i s';~v^*;~L^k[\epsilon] : t^*$
			\end{enumerate}
			
			By Case C1, H24, and the store typing relation,
			\begin{enumerate}[align=left, leftmargin=1\enumindent]
				\item[H26)] $\vdash s' : S'$
				\item[H27)] $S'_\text{inst}(i)_\text{roots} = \mathit{nil}^*$
			\end{enumerate}
			
			By H25, H26, H27 and the [Root] rule:		
			\begin{enumerate}[align=left, leftmargin=1\enumindent]
				\item[H28)] $\vdash_i^k s';v^*; L^k[\epsilon] : t^*$
			\end{enumerate}
			Case C1.6 is complete.

		\item $s;v^*; L^k[\op[\kappa]{i64.const} \continuationdelete] \hookrightarrow_i s;v^*; L^k[\op{trap}] \land \deletetrans(s, i, \kappa)~\text{undefined}$
		
			By the type checking of $\op{trap}$, H4, and \cref{lemma:substitution},
			\begin{enumerate}[align=left, leftmargin=1\enumindent]
				\item[H29)] $\vdash_i s;\; v^* ;\; L^k[\op{trap}] : t^*$
			\end{enumerate}
			By Case C1, H3, H29, and the [Root] rule, we conclude:
			\begin{enumerate}[align=left, leftmargin=1\enumindent]
				\item[H30)] $\vdash_i^k s;\; v^* ;\; L^k[\op{trap}] : t^*$
			\end{enumerate}
			Case C1.7 is complete.
			
		\item $s;v^*; L^k[\prompt{\textit{tf}}{e^*}] \hookrightarrow_i s';v^*; L^k[\op{block} \; \textit{tf} \; e^* \; \op{end} \; \op{prompt\_end}] \land s' = \prompttrans(s, i)$
		
			By Case C1.8:
			\begin{enumerate}[align=left, leftmargin=1\enumindent]
				\item[H31)] $\prompttrans(s, i) = s'~\text{where}~s' = s~\text{except}~s'_\text{inst}(i)_\text{pstack} \mapsto \text{push}(s_\text{inst}(i)_\text{pstack}, \{ \text{ctable} = \mathit{nil}^*, \text{root} = \mathit{nil} \})$
			\end{enumerate}
			
			By H4, H31, the type checking of $\op{prompt}$, and \cref{lemma:substitution},
			\begin{enumerate}[align=left, leftmargin=1\enumindent]
				\item[H32)] $\vdash_i s';\; v^* ;\; L^k[\op{block} \; \textit{tf} \; e^* \; \op{end} \; \op{prompt\_end}] : t^*$
			\end{enumerate}
			
			By H3, H31 and store typing:
			 \begin{enumerate}[align=left, leftmargin=1\enumindent]
				\item[H33)] $\vdash s' : S'$
				\item[H34)] $S'_\text{inst}(i)_\text{roots} = \mathit{nil}^*$
			\end{enumerate}
			
			By H32, H33, H34 and the [Root] rule:
			\begin{enumerate}[align=left, leftmargin=1\enumindent]
				\item[H35)] $\vdash_i^k s';\; v^* ;\; L^k[\op{block} \; \textit{tf} \; e^* \; \op{end} \; \op{prompt\_end}] : t^*$
			\end{enumerate}
			Case C1.8 is complete.
			
		\item $s;v^*; L^k[\promptend] \hookrightarrow_i s';v^*;\; L^k[\epsilon] \land  s' = \promptendtrans(s, i)$
		
			By Case C1.9:
			\begin{enumerate}[align=left, leftmargin=1\enumindent]
				\item[H36)] $\promptendtrans(s, i) = s'~\text{where}~s' = s~\text{except}~s'_\text{inst}(i)_\text{pstack} \mapsto \text{pop}(s_\text{inst}(i)_\text{pstack})$
			\end{enumerate}
			
			By H4, H36, the type checking of $\promptend$, and \cref{lemma:substitution},
			\begin{enumerate}[align=left, leftmargin=1\enumindent]
				\item[H37)] $\vdash_i s';\; v^* ;\; L^k[\epsilon] : t^*$
			\end{enumerate}
			
			By H3, H36 and store typing:
			 \begin{enumerate}[align=left, leftmargin=1\enumindent]
				\item[H38)] $\vdash s' : S'$
				\item[H39)] $S'_\text{inst}(i)_\text{roots} = \mathit{nil}^*$
			\end{enumerate}
			
			By H37, H38, H39 and the [Root] rule:
			\begin{enumerate}[align=left, leftmargin=1\enumindent]
				\item[H40)] $\vdash_i^k s';\; v^* ;\; L^k[\epsilon] : t^*$
			\end{enumerate}
			Case C1.9 is complete.
			
		\item $s;v^*; e^* \hookrightarrow_i s';v'^*;\; e'^*$ for some redex in standard WebAssembly.
		
			Since standard WebAssembly redexes do not modify the continuation tables, 
			$\vdash s' : S'$ where 
			$S'_\text{inst}(i)_\text{roots} = \mathit{nil}^*$.
			
			Therefore, by H4, the Preservation theorem of standard WebAssembly,
			and the [Root] rule:
			\begin{enumerate}[align=left, leftmargin=1\enumindent]
				\item[H41)] $\vdash_i^k s';\; v'^* ;\; e'^* : t^*$
			\end{enumerate}
			Case C1.10 is complete.
	\end{enumerate}
	
	\item $\exists p_r~\text{s.t.}~p_r = \max \{ p \mid S_\text{inst}(i)_\text{pstack}(p)_\text{root} \neq \mathit{nil} \}$ . In this case we also know:
	\begin{enumerate}[align=left, leftmargin=1\enumindent]
		\item[H42)] $\kappa_R = S_\text{inst}(i)_\text{pstack}(p_R)_\text{root}$
		\item[H43)] $\exists \tilde{t}^* \vdash_i s;~v^*;~e^* : \tilde{t}^*$
		\item[H44)] $t^* = S_\text{inst}(i)_\text{pstack}(p_R)_\text{ctable}(\kappa_R)$
	\end{enumerate}
	By H2 there are 10 subcases to consider:
	\begin{enumerate}[label*=\underline{.\arabic*}, align=left, leftmargin=1\enumindent]
		\item $s; v^*; \Lmax[\op[v]{i64.const} \control{h}] \rightsquigarrow_i s'; \epsilon ; \op[\kappa]{i64.const} \; (\mathbf{i64.const} \; v) \; (\mathbf{call} \; h) \; \mathbf{trap} \land (s',\kappa) = \controltrans(s, i, v^*, \Lmax)$
		
			There are two sub-sub cases, either $\fgetRoot(s, i) = \mathit{nil}$ or
			$\fgetRoot(s, i) = \kappa_0$ for some $\kappa_0$:
			
			\begin{enumerate}[label*=\underline{.\arabic*}, align=left, leftmargin=1\enumindent]
				\item $\fgetRoot(s, i) = \mathit{nil} \land p_r \neq 0$
					
					By Case C2.1 and C2.1.1:
					\begin{enumerate}[align=left, leftmargin=1\enumindent]
						\item[H45)] $\controltrans(s, i, v^*, \Lmax) = (s', \kappa) = (\fsetCont(\fsetRoot(s, i, \kappa), i, \kappa, \{ \mathrm{locals} = v_l^*, \mathrm{ctx} = \Lmax \}), \kappa)$
						\item[H46)] $\fgetCont(s, i, \kappa) = \mathit{nil}$
					\end{enumerate}
					
					By H3 and H45 we can type check $s'$:
					\begin{enumerate}[align=left, leftmargin=1\enumindent]
						\item[H47)] $\vdash s' : S'$
						\item[H48)] $p_r = \max \{ p \mid S'_\text{inst}(i)_\text{pstack}(p)_\text{root} \neq \mathit{nil} \}$
						\item[H49)] $S'_\text{inst}(i)_\text{pstack}(p_r)_\text{root} = \kappa_R$
						\item[H50)] $S'_\text{inst}(i)_\text{pstack}(p_r)_\text{ctable}(\kappa_R) = t^*$
					\end{enumerate}
				
					The new stack can be independently type checked. 
					By the $\op{control}$ type checking rule, 
					$S'_\text{inst}(i)_\text{func}(h) = \text{i64}~\text{i64} \to \epsilon$. Thus,
					\begin{enumerate}[align=left, leftmargin=1\enumindent]
						\item[H51)] $\vdash_i s';~\epsilon;~\op[\kappa]{i64.const}\op[v]{i64.const}\op[h]{call}\op{trap} : \epsilon$
					\end{enumerate}
					
					By H47, H48, H49, H50, H51 and the [Non-Root] rule:
					\begin{enumerate}[align=left, leftmargin=1\enumindent]
						\item[H52)] $\vdash_i^k s';~\epsilon;~\op[\kappa]{i64.const}\op[v]{i64.const}\op[h]{call}\op{trap} : t^*$
					\end{enumerate}
					Case C2.1.1 is complete.
					
				\item $\fgetRoot(s, i) = \kappa_0 \neq \mathit{nil}$
				
					By Case C2.1 and C2.1.2:
					\begin{enumerate}[align=left, leftmargin=1\enumindent]
						\item[H53)] $\controltrans(s, i, v^*, \Lmax) = (s', \kappa) = (\fsetCont(s, i, \kappa, \{ \mathrm{locals} = v_l^*, \mathrm{ctx} = \Lmax \}), \; \kappa)$
						\item[H54)] $\fgetCont(s, i, \kappa) = \mathit{nil}$
					\end{enumerate}
					
					By H3 and H53 we can type check $s'$:
					\begin{enumerate}[align=left, leftmargin=1\enumindent]
						\item[H55)] $\vdash s' : S'$
						\item[H56)] $p_r = \max \{ p \mid S'_\text{inst}(i)_\text{pstack}(p)_\text{root} \neq \mathit{nil} \}$
						\item[H57)] $S'_\text{inst}(i)_\text{pstack}(p_r)_\text{root} = \kappa_R$
						\item[H58)] $S'_\text{inst}(i)_\text{pstack}(p_r)_\text{ctable}(\kappa_R) = t^*$
					\end{enumerate}
				
					The new stack can be independently type checked. 
					By the $\op{control}$ type checking rule, 
					$S'_\text{inst}(i)_\text{func}(h) = \text{i64}~\text{i64} \to \epsilon$. Thus,
					\begin{enumerate}[align=left, leftmargin=1\enumindent]
						\item[H59)] $\vdash_i s';~\epsilon;~\op[\kappa]{i64.const}\op[v]{i64.const}\op[h]{call}\op{trap} : \epsilon$
					\end{enumerate}
					
					By H55, H56, H57, H58, H59 and the [Non-Root] rule:
					\begin{enumerate}[align=left, leftmargin=1\enumindent]
						\item[H60)] $\vdash_i^k s';~\epsilon;~\op[\kappa]{i64.const}\op[v]{i64.const}\op[h]{call}\op{trap} : t^*$
					\end{enumerate}
					Case C2.1.2 is complete.
					
			\end{enumerate}

		\item $s; v^*; \Lmax[\op[\kappa]{i64.const} \op[v]{i64.const} \restore] \rightsquigarrow_i s';\; v'^* ;\; {\Lmax}' [(\mathbf{i64.const} \; v)] \land (s', v^*, {\Lmax}') = \restoretrans(s, i, \kappa)$

			By C2.2:
			\begin{enumerate}[align=left, leftmargin=1\enumindent]
				\item[H61)] $\fgetRoot(s, i) = \kappa_0 \neq \mathit{nil}$
			\end{enumerate}
			 
			There are two sub-sub cases, either $\kappa = \kappa_0$ or
			$\kappa \neq \kappa_0$:
			
			\begin{enumerate}[label*=\underline{.\arabic*}, align=left, leftmargin=1\enumindent]
				\item $\kappa = \kappa_0$
					
					By C2.2.1:
					\begin{enumerate}[align=left, leftmargin=1\enumindent]
						\item[H62)] $\restoretrans(s, i, \kappa) = (s', v'^*, {\Lmax}') = (\fsetRoot(\fsetCont(s, i, \kappa, \mathit{nil}), i, \mathit{nil}), \; \fgetCont(s, i, \kappa)_\text{locals}, \; \fgetCont(s, i, \kappa)_\text{ctx})$
					\end{enumerate}
					
					By H3 and H62 we can type check s':
					\begin{enumerate}[align=left, leftmargin=1\enumindent]
						\item[H63)] $\vdash s' : S'$
					\end{enumerate}
					
					There are two sub-sub-sub cases, either $p_r = 0$ or $p_r \geq 1$:
					\begin{enumerate}[label*=\underline{.\arabic*}, align=left, leftmargin=1\enumindent]
					
						\item $p_r = 0$
						
							By C2.2.1.1 and H42,
							\begin{enumerate}[align=left, leftmargin=1\enumindent]
								\item[H64)] $\kappa = \kappa_0 = \kappa_R = S_\text{inst}(i)_\text{pstack}(p_R)_\text{root}$
							\end{enumerate}
							
							By C2.2.1.1, H62, and H63:
							\begin{enumerate}[align=left, leftmargin=1\enumindent]
								\item[H65)] $S'_\text{inst}(i)_\text{roots} = \mathit{nil}^*$
							\end{enumerate}
							
							By C2.2.1.1, H64, and \cref{lemma:substitution}:
							\begin{enumerate}[align=left, leftmargin=1\enumindent]
								\item[H66)] $\vdash_i s';~v'^*;~{\Lmax}'[\op[v]{i64.const}] : t^*$
							\end{enumerate}
							
							By H63, H65, H66, and the [Root] rule:
							\begin{enumerate}[align=left, leftmargin=1\enumindent]
								\item[H67)] $\vdash_i^k s';~v'^*;~{\Lmax}'[\op[v]{i64.const}] : t^*$
							\end{enumerate}
							Case C2.2.1.1 is complete.
						
						\item $p_r \geq 1$
							
							By C2.2.1.2 and H63:
							\begin{enumerate}[align=left, leftmargin=1\enumindent]
								\item[H68)] $p_r = \max \{ p \mid S'_\text{inst}(i)_\text{pstack}(p)_\text{root} \neq \mathit{nil} \}$
								\item[H69)] $S'_\text{inst}(i)_\text{pstack}(p_r)_\text{root} = \kappa_R$
								\item[H70)] $S'_\text{inst}(i)_\text{pstack}(p_r)_\text{ctable}(\kappa_R) = t^*$
							\end{enumerate}
							
							The stack we are switching to can be type checked by H3 and \cref{lemma:substitution}:
							\begin{enumerate}[align=left, leftmargin=1\enumindent]
								\item[H71)] $\exists \tilde{t}'^*, \vdash_i s';~v'^*;~{\Lmax}'[\op[v]{i64.const}] : \tilde{t}'^*$
							\end{enumerate}
							
							By H63, H68, H69, H70, H71, and the [Non-Root] rule:
							\begin{enumerate}[align=left, leftmargin=1\enumindent]
								\item[H72)] $\vdash_i^k s';~v'^*;~{\Lmax}'[\op[v]{i64.const}] : t^*$
							\end{enumerate}
							
							Case C2.2.1.2 is complete.

					\end{enumerate}
					
				\item $\kappa \neq \kappa_0$
				
					By C2.2.2:
					\begin{enumerate}[align=left, leftmargin=1\enumindent]
						\item[H73)] $\restoretrans(s, i, \kappa) = (s', v'^*, {\Lmax}') = (\fsetCont(s, i, \kappa, \mathit{nil}), \; \fgetCont(s, i, \kappa)_\text{locals}, \; \fgetCont(s, i, \kappa)_\text{ctx})$
					\end{enumerate}
					
					By H3 and H73 we can type check s':
					\begin{enumerate}[align=left, leftmargin=1\enumindent]
						\item[H74)] $\vdash s' : S'$
						\item[H75)] $p_r = \max \{ p \mid S'_\text{inst}(i)_\text{pstack}(p)_\text{root} \neq \mathit{nil} \}$
						\item[H76)] $S'_\text{inst}(i)_\text{pstack}(p_r)_\text{root} = \kappa_R$
						\item[H77)] $S'_\text{inst}(i)_\text{pstack}(p_r)_\text{ctable}(\kappa_R) = t^*$
					\end{enumerate}
					
					The stack we are switching to can be type checked by H3 and \cref{lemma:substitution}:
					\begin{enumerate}[align=left, leftmargin=1\enumindent]
						\item[H78)] $\exists \tilde{t}'^*, \vdash_i s';~v'^*;~{\Lmax}'[\op[v]{i64.const}] : \tilde{t}'^*$
					\end{enumerate}
					
					By H74, H75, H76, H77, H78, and the [Non-Root] rule:
					\begin{enumerate}[align=left, leftmargin=1\enumindent]
						\item[H79)] $\vdash_i^k s';~v'^*;~{\Lmax}'[\op[v]{i64.const}] : t^*$
					\end{enumerate}
					
			\end{enumerate}

		\item $s; v^*; \Lmax[\op[\kappa]{i64.const} \op[v]{i64.const} \restore] \rightsquigarrow_i s;\; v^* ;\; \op{trap} \land \restoretrans(s, i, \kappa)~\text{undefined}$
					
			By the type checking of $\op{trap}$,
			\begin{enumerate}[align=left, leftmargin=1\enumindent]
				\item[H80)] $\vdash_i s;\; v^* ;\; \op{trap} : \epsilon$
			\end{enumerate}
			
			By Case C2, H3, H42, H44, H80, and the [Non-Root] rule, we conclude:
			\begin{enumerate}[align=left, leftmargin=1\enumindent]
				\item[H81)] $\vdash_i^k s;\; v^* ;\; \op{trap} : t^*$
			\end{enumerate}
			Case C2.3 is complete.
		
		\item $s;v^*; L^k[\op[\kappa]{i64.const} \continuationcopy] \hookrightarrow_i s';v^*; L^k[\op[\kappa']{i64.const}] \land (s', \kappa') = \copytrans(s, i, \kappa)$
					
			By Case C2.4:
			\begin{enumerate}[align=left, leftmargin=1\enumindent]
				\item[H82)] $\copytrans(s, i, \kappa) = (s', \kappa') = (\fsetCont(s, i, \kappa', \fgetCont(s, i, \kappa)), \kappa')$
				\item[H83)] $\fgetCont(s, i, \kappa') = \mathit{nil}$.
			\end{enumerate}
			
			By H43, Case C2.4, and \cref{lemma:substitution},
			\begin{enumerate}[align=left, leftmargin=1\enumindent]
				\item[H84)] $\vdash_i s';~v^*;~L^k[\op[\kappa']{i64.const}] : \tilde{t}^*$
			\end{enumerate}
			
			By H3, there exists a type for the stack that is being copied:
			\begin{enumerate}[align=left, leftmargin=1\enumindent]
				\item[H85)] $\exists t'^*,\; S_\text{inst}(i)_\text{pstack}(0)_\text{ctable}(\kappa) = t'^*$
			\end{enumerate}
			
			By Case C2, H82, H83, H85, and the store typing relation:
			\begin{enumerate}[align=left, leftmargin=1\enumindent]
				\item[H86)] $\vdash s' : S'$
				\item[H87)] $p_r = \max \{ p \mid S'_\text{inst}(i)_\text{pstack}(p)_\text{root} \neq \mathit{nil} \}$
				\item[H88)] $S'_\text{inst}(i)_\text{pstack}(p_r)_\text{root} = \kappa_R$
				\item[H89)] $S'_\text{inst}(i)_\text{pstack}(p_r)_\text{ctable}(\kappa_R) = t^*$
			\end{enumerate}
			
			By H84, H86, H87, H88, H89 and the [Non-Root] rule:		
			\begin{enumerate}[align=left, leftmargin=1\enumindent]
				\item[H90)] $\vdash_i^k s';v^*; L^k[\op[\kappa']{i64.const}] : t^*$
			\end{enumerate}
			
			Case C2.4 is complete.

		\item $s;v^*; L^k[\op[\kappa]{i64.const} \continuationcopy] \hookrightarrow_i s;v^*; L^k[\op{trap}] \land \copytrans(s, i, \kappa)~\text{undefined}$
					
			By the type checking of $\op{trap}$, H43, and \cref{lemma:substitution},
			\begin{enumerate}[align=left, leftmargin=1\enumindent]
				\item[H91)] $\vdash_i s;\; v^* ;\; L^k[\op{trap}] : \tilde{t}^*$
			\end{enumerate}
			
			By Case C2, H3, H42, H44, H91, and the [Non-Root] rule, we conclude:
			\begin{enumerate}[align=left, leftmargin=1\enumindent]
				\item[H92)] $\vdash_i^k s;\; v^* ;\; L^k[\op{trap}] : t^*$
			\end{enumerate}
			Case C2.5 is complete.
		
		\item $s;v^*; L^k[\op[\kappa]{i64.const} \continuationdelete] \hookrightarrow_i s';v^*; L^k[\epsilon] \land s' = \deletetrans(s, i, \kappa)$
						
			By Case C2.6:
			\begin{enumerate}[align=left, leftmargin=1\enumindent]
				\item[H93)] $\deletetrans(s, i, \kappa) = s' = \fsetCont(s, i, \kappa, \mathit{nil})$
			\end{enumerate}
			
			By H43, Case C2.6, and \cref{lemma:substitution},
			\begin{enumerate}[align=left, leftmargin=1\enumindent]
				\item[H94)] $\vdash_i s';~v^*;~L^k[\epsilon] : \tilde{t}^*$
			\end{enumerate}
			
			By Case C2, H93, and the store typing relation:
			\begin{enumerate}[align=left, leftmargin=1\enumindent]
				\item[H95)] $\vdash s' : S'$
				\item[H96)] $p_r = \max \{ p \mid S'_\text{inst}(i)_\text{pstack}(p)_\text{root} \neq \mathit{nil} \}$
				\item[H97)] $S'_\text{inst}(i)_\text{pstack}(p_r)_\text{root} = \kappa_R$
				\item[H98)] $S'_\text{inst}(i)_\text{pstack}(p_r)_\text{ctable}(\kappa_R) = t^*$
			\end{enumerate}
			
			By H94, H95, H96, H97, H98 and the [Non-Root] rule:		
			\begin{enumerate}[align=left, leftmargin=1\enumindent]
				\item[H99)] $\vdash_i^k s';v^*; L^k[\op[\kappa']{i64.const}] : t^*$
			\end{enumerate}
			
			Case C2.6 is complete.
		
		\item $s;v^*; L^k[\op[\kappa]{i64.const} \continuationdelete] \hookrightarrow_i s;v^*; L^k[\op{trap}] \land \deletetrans(s, i, \kappa)~\text{undefined}$

			By the type checking of $\op{trap}$, H43, and \cref{lemma:substitution},
			\begin{enumerate}[align=left, leftmargin=1\enumindent]
				\item[H100)] $\vdash_i s;\; v^* ;\; L^k[\op{trap}] : \tilde{t}^*$
			\end{enumerate}
			
			By Case C2, H3, H42, H44, H100, and the [Non-Root] rule, we conclude:
			\begin{enumerate}[align=left, leftmargin=1\enumindent]
				\item[H101)] $\vdash_i^k s;\; v^* ;\; L^k[\op{trap}] : t^*$
			\end{enumerate}
			
			Case C2.7 is complete.

		\item $s;v^*; L^k[\prompt{\textit{tf}}{e^*}] \hookrightarrow_i s';v^*; L^k[\op{block} \; \textit{tf} \; e^* \; \op{end} \; \op{prompt\_end}] \land s' = \prompttrans(s, i)$
					
			By Case C2.8:
			\begin{enumerate}[align=left, leftmargin=1\enumindent]
				\item[H102)] $\prompttrans(s, i) = s'~\text{where}~s' = s~\text{except}~s'_\text{inst}(i)_\text{pstack} \mapsto \text{push}(s_\text{inst}(i)_\text{pstack}, \{ \text{ctable} = \mathit{nil}^*, \text{root} = \mathit{nil} \})$
			\end{enumerate}
			
			By H43, H102, the type checking of $\op{prompt}$, and \cref{lemma:substitution},
			\begin{enumerate}[align=left, leftmargin=1\enumindent]
				\item[H103)] $\vdash_i s';\; v^* ;\; L^k[\op{block} \; \textit{tf} \; e^* \; \op{end} \; \op{prompt\_end}] : \tilde{t}^*$
			\end{enumerate}
			
			By H3, H102 and store typing:
			 \begin{enumerate}[align=left, leftmargin=1\enumindent]
				\item[H104)] $\vdash s' : S'$
				\item[H105)] $p_r = \max \{ p \mid S'_\text{inst}(i)_\text{pstack}(p)_\text{root} \neq \mathit{nil} \}$
				\item[H106)] $S'_\text{inst}(i)_\text{pstack}(p_r)_\text{root} = \kappa_R$
				\item[H107)] $S'_\text{inst}(i)_\text{pstack}(p_r)_\text{ctable}(\kappa_R) = t^*$
			\end{enumerate}
			
			By H103, H104, H105, H106, H107, and the [Non-Root] rule:
			\begin{enumerate}[align=left, leftmargin=1\enumindent]
				\item[H108)] $\vdash_i^k s';\; v^* ;\; L^k[\op{block} \; \textit{tf} \; e^* \; \op{end} \; \op{prompt\_end}] : t^*$
			\end{enumerate}
			
			Case C2.8 is complete.
			
		\item $s;v^*; L^k[\promptend] \hookrightarrow_i s';v^*;\; L^k[\epsilon] \land  s' = \promptendtrans(s, i)$

			The administrative store typing of $\promptend$ implies:
			\begin{enumerate}[align=left, leftmargin=1\enumindent]
				\item[H109)] $S_\text{inst}(i)_\text{pstack}(0) = \mathit{nil}$
			\end{enumerate}
			
			By Case C2.9:
			\begin{enumerate}[align=left, leftmargin=1\enumindent]
				\item[H110)] $\promptendtrans(s, i) = s'~\text{where}~s' = s~\text{except}~s'_\text{inst}(i)_\text{pstack} \mapsto \text{pop}(s_\text{inst}(i)_\text{pstack})$
			\end{enumerate}
			
			By H43, H110, the type checking of $\promptend$, and \cref{lemma:substitution},
			\begin{enumerate}[align=left, leftmargin=1\enumindent]
				\item[H111)] $\vdash_i s';\; v^* ;\; L^k[\epsilon] : \tilde{t}^*$
			\end{enumerate}
			
			By H3, H109, H110 and store typing:
			 \begin{enumerate}[align=left, leftmargin=1\enumindent]
				\item[H112)] $\vdash s' : S'$
				\item[H113)] $p_r - 1 = \max \{ p \mid S'_\text{inst}(i)_\text{pstack}(p)_\text{root} \neq \mathit{nil} \}$
				\item[H114)] $S'_\text{inst}(i)_\text{pstack}(p_r - 1)_\text{root} = \kappa_R$
				\item[H115)] $S'_\text{inst}(i)_\text{pstack}(p_r - 1)_\text{ctable}(\kappa_R) = t^*$
			\end{enumerate}
			
			By H111, H112, H113, H114, H115, and the [Non-Root] rule:
			\begin{enumerate}[align=left, leftmargin=1\enumindent]
				\item[H116)] $\vdash_i^k s';\; v^* ;\; L^k[\epsilon] : t^*$
			\end{enumerate}
			
			Case C2.9 is complete.
			
		\item $s;v^*; e^* \hookrightarrow_i s';v'^*;\; e'^*$ for some redex in standard WebAssembly.
					
			Since standard WebAssembly redexes do not modify the continuation tables,
			\begin{enumerate}[align=left, leftmargin=1\enumindent]
				\item[H117)] $\vdash s' : S'$
				\item[H118)] $S'_\text{inst}(i)_\text{pstack} = S_\text{inst}(i)_\text{pstack}$
			\end{enumerate}

			Therefore, by H43 and the Preservation theorem of standard WebAssembly:
			\begin{enumerate}[align=left, leftmargin=1\enumindent]
				\item[H119)] $\vdash_i s';\; v'^* ;\; e'^* : \tilde{t}^*$
			\end{enumerate}
			
			By H117, H118, H119, and the [Non-Root] rule:
			\begin{enumerate}[align=left, leftmargin=1\enumindent]
				\item[H120)] $\vdash_i^k s';\; v'^* ;\; e'^* : t^*$
			\end{enumerate}
			Case C2.10 is complete.
	\end{enumerate}
\end{enumerate}
\end{proof}

}

\begin{theorem}[$\rightsquigarrow$ Progress]
	If $\vdash_i^k s; v^*; e^* : t^*$, then either $e^* = v'^*$ or $e^* = \textbf{trap}$ or $s; v^*; e^* \rightsquigarrow_i s'; v'^*; e'^*$.
\end{theorem}

\begin{proof}
	The reduction rules for 
	$\op{control}$, $\restore$, $\continuationcopy$, 
	$\continuationdelete$, and $\op{prompt}$ can be 
	trivially checked to cover the space of possible 
	well-typed configurations $s; v^*; e^*$. 
	Thus, all redexes for these instructions are guaranteed to 
	take a step with $\rightsquigarrow$ to a new configuration (possibly a step to a $\op{trap}$).
	
	The non-trivial case is $\promptend$. 
	Suppose $\vdash_i^k s; v^*; L^k[\promptend] : t^*$. 
	From the reduction rule of $\promptend$, 
	a step $s; v^*; L^k[\promptend] \rightsquigarrow s; v^*; L^k[\epsilon]$ 
	will occur if $\fgetRoot(s, i) = \mathit{nil}$. 
	
	However, no such step occurs if $\fgetRoot(s, i) \neq \mathit{nil}$. 
	We thus want to show that $\fgetRoot(s, i) \neq \mathit{nil}$ 
	is in contradiction with $\vdash_i^k s; v^*; L^k[\promptend] : t^*$, 
	implying that this case cannot occur. By the [Prompt-End] rule and 
	$\vdash_i^k s; v^*; L^k[\promptend] : t^*$ we can deduce that $\vdash s : S$
	and $S_\text{inst}(i)_\text{pstack}(0)_\text{root} = \mathit{nil}$.
	From this and the store typing rules, 
	we find that $\fgetRoot(s, i) = \mathit{nil}$, which is a contradiction.
\end{proof}

\fi

\end{document}